\documentclass[a4paper,reqno,11pt]{amsart} 

\usepackage[T1]{fontenc} 
\usepackage[utf8]{inputenc}
\usepackage[english]{babel} 
\usepackage{newpxtext} 

\usepackage{scalerel}

\usepackage{fullpage}
\usepackage{amsmath, amsthm}
\usepackage{amscd}
\usepackage{enumerate}
\usepackage{float,amsmath,amssymb,mathrsfs,bm,multirow,graphics}
\usepackage[percent]{overpic}
\usepackage[numbers,sort&compress]{natbib}
\usepackage{xcolor}
\usepackage{todonotes}
\usepackage{amsaddr}
\usepackage{mathtools}
\usepackage{subcaption}
\usepackage{mathtools}
\usepackage{stackengine} 
\usepackage{ytableau}
\usepackage{tikz}

\usepackage{stmaryrd} 
\usepackage{dsfont}

\newtheorem{cj}{Conjecture}
\newtheorem*{cj*}{Conjecture}
\newtheorem{theorem}{Theorem}
\newtheorem{proposition}{Proposition}[section]
\newtheorem{lemma}[proposition]{Lemma}
\newtheorem{corollary}[proposition]{Corollary}
\newtheorem{definition}[proposition]{Definition}

\newtheorem{remark}[proposition]{Remark}

\newtheorem{example}[proposition]{Example}

 \newcommand\ben{\begin{equation*}}
 \newcommand\ebn{\end{equation*}}
 \newcommand\beq{\begin{equation}}
 \newcommand\eeq{\end{equation}}
 \newcommand\lb{\left(}
  \newcommand\rb{\right)} 
\numberwithin{equation}{section}

 \newcommand{\np}{d}
 \newcommand{\Summed}{n}

 \newcommand{\NB}{\text{NB}_\lambda}
 \newcommand{\SYT}{\text{SYT}_\lambda}
 \newcommand{\RSYT}{\text{RSYT}^\lambda_\varsigma}
 \newcommand{\CSYT}{\text{CSYT}^\lambda_\varsigma}
 \newcommand{\idpts}{s_{\varsigma}}
 \newcommand{\idpta}{a_{\varsigma}}

\ytableausetup{centertableaux,smalltableaux}

\allowdisplaybreaks

\usepackage[colorlinks,psdextra,pdfencoding=auto]{hyperref}
\hypersetup{colorlinks,linkcolor=red,filecolor=green,urlcolor=blue,citecolor=blue} 

\usepackage{etoolbox} 
    \newcommand{\addQEDstyle}[2]{\AtBeginEnvironment{#1}{\pushQED{\qed}\renewcommand{\qedsymbol}{#2}}
    \AtEndEnvironment{#1}{\popQED}} 
    \addQEDstyle{remark}{$\circ$}
    \addQEDstyle{remark*}{$\circ$} 
	\addQEDstyle{example}{$\circ$}
    \addQEDstyle{example*}{$\circ$} 

\title{Degenerate conformal blocks for the $W_3$ algebra at $c=2$ and connection probabilities in the triple dimer model}

\author{Augustin Lafay$^*$, Ian Le$^\dagger$ and Julien Roussillon$^*$}

\address{$^*$\quad Department of Mathematics and Systems Analysis, 
\\ \phantom{$^*$}\quad P.O. Box 11100, FI-00076 Aalto, Finland}

\address{$^\dagger$\quad Mathematical Sciences Institute, 
\\ \phantom{$^\dagger$}\quad Australian National University, Australia}

\email{augustin.lafay@aalto.fi}
\email{ian.le@anu.edu.au}
\email{julien.roussillon@aalto.fi}

\begin{document}
\maketitle
\begin{abstract}
    We study a homogeneous system of $d+8$ linear partial differential equations (PDEs) in $d$ variables arising from two-dimensional Conformal Field Theories (CFTs) with a $W_3$-symmetry algebra. In the CFT context, $d$ PDEs are third-order and correspond to the null-state equations, whereas the remaining 8 PDEs (five being second-order and three being first-order) correspond to the $W_3$ global Ward identities. In the case of central charge $c=2$, we construct a subspace of the space of all solutions which grow no faster than a power law. We call this subspace the space of $W_3$ conformal blocks, and we provide a basis expressed in terms of Specht polynomials associated with column-strict, rectangular Young tableaux with three columns. The dimension of this space is a Kostka number which coincides with CFT predictions, hence we conjecture that it exhausts the space of all solutions having a power law bound. Moreover, we prove that the space of $W_3$ conformal blocks is an irreducible representation of a certain diagram algebra defined from $\mathfrak{sl}_3$ webs that we call Kuperberg algebra. Finally, we prove a formula relating the $W_3$ conformal blocks at $c=2$ we constructed to Kenyon and Shi's scaling limits of connection probabilities in the triple dimer model. For more general central charges, we expect that $W_3$ conformal blocks are related to scaling limits of probabilities in lattice models based on $\mathfrak{sl}_3$ webs. 
\end{abstract}

{\hypersetup{linkcolor=black}
\setcounter{tocdepth}{2}
\tableofcontents}

\section{Introduction}




\subsection*{Introduction} Since its inception in \cite{BPZ}, Conformal Field Theory (CFT) has provided a formidable playground to both theoretical physics and mathematics. To this day, there is no general agreement on what the definition of a CFT is. A powerful approach to CFT developed in \cite{BPZ} and generally used in the physics literature is called conformal bootstrap. In this approach, one postulates a set of conformally covariant local fields called \emph{primary fields}, along with a rule called \emph{operator product expansion} (OPE) which expresses the product of two primary fields within a correlation function as a sum over a subset of primary fields called \emph{spectrum}. The coefficients in the OPE contain the so-called \emph{structure constants} which are closely related to the three-point correlation functions. Given the knowledge of the spectrum and the structure constants, higher-point correlation functions are then determined recursively. 

In the conformal bootstrap approach, correlation functions of $N>3$ primary fields on the Riemann sphere conjecturally decompose into products of structure constants and universal quantities called \emph{conformal blocks}. These are special functions which arise solely from representation theory of the underlying symmetry algebra of the CFT. A special emphasis is given to the four-point conformal blocks, since the four-point correlation function carries a substantial amount of information about the model. For instance, it has to satisfy the so-called crossing symmetry equations which are consequences of the associativity of the operator product expansion between primary fields of the model. More generally, for a review of the conformal bootstrap approach to CFT in the physics literature, the reader is referred to \cite{DMS,Rib}. 

Conformal blocks are not well understood in general. For the celebrated class of conformal blocks appearing in Liouville theory, until recently, only formal power series expansions were available in the physics literature and the AGT correspondence \cite{AGT} provided a conjectural closed-form formula for the coefficients of the series. In a remarkable series of recent works culminating in \cite{GKRV2020,GKRV2021}, those conformal blocks have been constructed rigorously.

In the special case where $N \geq 1$ primary fields are degenerate, that is, they have a descendant field which is a so-called null-field, the conformal blocks of the Virasoro algebra satisfy a system of $N$ partial differential equations called Virasoro-BPZ equations. For the simplest nontrivial degenerate fields, the resulting PDEs are second order \cite{BPZ,KRV19}. 

Besides their key role in the conformal bootstrap approach to CFTs, Virasoro conformal blocks of degenerate fields are closely related to the scaling limits of critical lattice models such as the Ising or the percolation model. For instance, imposing suitable boundary conditions on a simply connected domain typically leads to the emergence of possibly multiple interfaces connecting boundary points. In the scaling limit, such interfaces become a system of conformally invariant random curves called Schramm-Loewner evolutions (SLE) \cite{S, Sc11}. In the case where multiple interfaces are present, they are described by the so-called partitions functions which are constrained by a martingale property \cite{BBK,Gr,Du, KL07, Pe19} which amounts to the second-order system of BPZ equations mentioned above. Moreover, in this context, interfaces can join in different patterns known as \emph{link patterns}. A \emph{connection probability} is the probability of obtaining a given link pattern, and pure partition functions are expected to provide probability weights for the scaling limit of connection probabilities \cite{BBK, KP2}. An important example for the present paper is the case of the double dimer model, where the relation between conformal blocks and such probability weights was made explicit \cite{KW1}. In a broader context, the relations between these lattice models and CFT were first conjectured in the physics literature \cite{DS,Ca,Ja}. These predictions were based on a Coulomb gas argument, which models the scaling limit of loop models using the compactified imaginary Liouville theory. This CFT has recently been rigorously constructed in \cite{GKR}.

More generally, in addition to curves joining boundary points according to a given link pattern, interfaces in lattice models also include collections of loops within the domain. Such link patterns and loops have a natural interpretation in terms of invariants of $U_q(\mathfrak{sl}_2)$. Higher-rank generalizations of link patterns, called \emph{webs}, have been introduced in representation theory \cite{Ku,CKM}. Moreover, one of the present authors recently realized in \cite{LGJ1} at a physics level of rigor that there exist macroscopic objects in some lattice models -- such as interfaces in the Potts model -- which are described by these webs (see also \cite{FLL,DKS,Ta1,Ta2} in the mathematics literature on dimer models).  

More recently, Kenyon and Shi rigorously constructed connection probabilities in the triple dimer model \cite{KS}, where the connection patterns are given by webs. Their result was one of the main motivations for the present work. Indeed, by analogy with the case of the double dimer model, it is natural to expect that Kenyon and Shi's connection probabilities can be expressed in terms of higher-rank analogues of Virasoro conformal blocks. The main contribution of this article is to construct such conformal blocks, understand their CFT meaning, and rigorously relate them to Kenyon and Shi's connection probabilities.

We now heuristically describe what kind of CFT naturally arises in this context. One of the authors of the present article argued in \cite{LGJ2} at a theoretical physics level of rigor that the scaling limits of models based on $\mathfrak{sl}_3$ webs are described by a compactified imaginary $\mathfrak{sl}_3$ Toda field theory -- a higher rank analog of Liouville theory -- where the central charge satisfies $c \in(-\infty,2]$. This suggests that the corresponding observables in these critical lattice models cannot be described solely using SLE. Finally, although the compactified imaginary $\mathfrak{sl}_3$ Toda field theory is not understood rigorously, the path integral formulation of the real $\mathfrak{sl}_3$ Toda field theory has been rigorously studied in recent works \cite{CH,CRV}. This CFT belongs to a large class of theories whose symmetry algebra strictly contains the Virasoro algebra. Such extended symmetry algebras are generally called $W$-algebras \cite{BS}. 

\subsection*{Brief description of our results} In this paper, we study a system of equations arising from CFTs whose symmetry algebra is the $W_3$ algebra, which is also called $WA_2$ algebra \cite{Zam85}. Just as the simplest nontrivial Virasoro-degenerate fields lead to second-order PDEs, the simplest nontrivial $W_3$-degenerate fields give rise to third-order PDEs (see for instance \cite{FL}). We introduce a system comprising $d$ such third-order PDEs, which we call $W_3$-BPZ equations, along with five second-order PDEs and three first-order PDEs. Physically, the first and second-order PDEs are the global Ward identities associated with the currents $T$ and $\mathcal W$ which encode the $W_3$ symmetry algebra of the model. Solutions of this system will be referred to as $W_3$ conformal blocks.

The present paper describes three classes of results which are analytic, algebraic and probabilistic in nature. The analytic results consist of an explicit construction of $W_3$ conformal blocks in the special case of central charge $c=2$. More precisely, we construct a subspace of solutions of the $d+8$ PDEs, whose elements are expressed in terms of the celebrated Specht polynomials associated with rectangular Young diagrams with $3$ columns. Although we were not able to prove that this space exhausts the space of all solutions, we expect that this is the case, since its dimension is the one predicted by the fusion rules of CFT. As for the algebraic results, it is known that the solution space of the Virasoro-BPZ equations forms an irreducible representation of the Temperley-Lieb algebra for irrational central charge \cite{FP20, KP1}. In this paper, we prove a higher-rank analog for $c=2$. More precisely, at this value of the central charge, we prove that the solution space of the $W_3$-BPZ equations forms an irreducible representation of a certain diagram algebra denoted $\text{K}^\varsigma$ defined from $\mathfrak sl_3$ webs \cite{Ku} which we call Kuperberg algebra. Finally, our probabilistic results are a complete characterization of the scaling limits of Kenyon and Shi's connection probabilities in the triple dimer model in terms of the $W_3$ conformal blocks at $c=2$. This result is a higher-rank analogue of the one relating connection probabilities in the double dimer model and multiple SLE partition functions \cite{KW1}. 

Although the main results of this paper focus on the central charge $c=2$, we expect that our three classes of results have analogues for arbitrary central charge. Specifically, we conjecture that $W_3$ conformal blocks at any central charge provide amplitudes for the scaling limits of probabilities in lattice models based on $\mathfrak{sl}_3$ webs. If a stochastic description of $\mathfrak{sl}_3$ webs were to be constructed, we would expect the $W_3$-BPZ equations to play a central role in the case of multiple boundary points.

\subsection{Organization of the paper}

The paper is organized as follows. In Section \ref{section2}, we introduce the space of functions satisfying the $d+8$ PDEs of interest, and we provide a more precise formulation of our main results. The following sections are devoted to the proofs of the algebraic, probabilistic, and analytic results in order. Specifically, in Section \ref{sec:algebra} we study the algebraic structures associated with the space of $W_3$ conformal blocks at $c=2$. Our probabilistic results, which relate Kenyon and Shi's connection probabilities and the $W_3$ conformal blocks at $c=2$, are established in section \ref{sec:dimer}. In Section \ref{sectionprooftheorem1}, we prove our analytic results, which construct the space of $W_3$ conformal blocks at $c=2$. Finally, a derivation of the $d+8$ PDEs of interest is described at a theoretical physics level of rigor in the Appendix.

We emphasize that Sections \ref{sec:algebra}, \ref{sec:dimer} and \ref{sectionprooftheorem1} are almost entirely independent. We chose to present them in this order for two reasons. First, the Section \ref{sec:dimer} on the triple dimer model relies on algebraic results from \cite{Ku} which we recall in Section \ref{sec:algebra}. Second, Section \ref{sectionprooftheorem1} assumes linear independence of the $W_3$-conformal blocks, which is proved in Section \ref{sec:algebra}.



\subsection{Acknowledgements}
We thank Eveliina Peltola and Richard Kenyon for illuminating discussions which motivated us to start this project. We also thank Richard Kenyon and Haolin Shi for helping us understanding their work, and Eveliina Peltola for her useful comments on the draft. A.L. is supported by the Academy of Finland grant number 340461, entitled “Conformal invariance in planar random geometry”. J.R. is supported by the Academy of Finland Centre of Excellence Programme grant number 346315 entitled "Finnish centre of excellence in Randomness and STructures (FiRST)". 

\section{Formulation of the main results} \label{section2}

\subsection{Notations} 
Throughout the paper, the following notations will be used. Let $d \in \mathbb N$ and let $\mathcal H_d$ be the configuration space
\begin{equation}
    \mathcal H_d = \{ (x_1,\cdots,x_d) \in \mathbb R^d : x_1 < \cdots < x_d \}.
\end{equation}
Moreover, let $n=0 \; \text{mod} \; 3$ and let $\varsigma=(s_1,\cdots,s_d)$ be such that $\sum_{i=1}^d s_i = n$, with $s_i=1$ or $s_i=2$ for all $i=1,\cdots,d$. We also introduce $q_i = (-1)^{s_i+1}$ for all $i=1,\cdots,d$.
Finally, we define the partition $\pi$ of $n$ by
\begin{equation} \label{defpi}
    \pi=\left(\frac{n}{3},\frac{n}{3},\frac{n}{3}\right),
\end{equation} and denote by $\Bar{\pi}=(3^{(\frac{n}{3})})$ its conjugate partition. The main results of this paper will focus exclusively on partitions of the form $\pi$.

\subsection{Definitions}

In this section, we introduce in Definition \ref{defSvarsigmac} the $\mathbb R$-vector space denoted $\mathcal S_\varsigma^{(c)}$ of all functions $f: \mathcal H_d \to \mathbb R$ satisfying a homogeneous system of $d+8$ linear PDEs and a power law bound. We also introduce in Definition \ref{defconformalblocks} the space $\mathcal C^{(c=2)}_\varsigma$ of $W_3$-conformal blocks at central charge $c=2$ in terms of the celebrated Specht polynomials. These two definitions form the core objects of the paper upon which the main results are built.

\subsubsection{The space $\mathcal S_\varsigma^{(c)}$}

Let $\beta\in \mathbb R_{> 0}$ and define $c(\beta)$ and $h(\beta)$ such that
\begin{align}\label{relationcandbeta}
    c = 2 - 24 \lb \beta - \frac{1}{\beta} \rb^2,\qquad h(\beta) = \frac{4\beta^2}3 - 1.
\end{align} 
For convenience, the dependence of $c$ and $h$ on $\beta$ will be omitted. In CFT, $c$ and $h$ correspond to the central charge and the conformal dimension of the degenerate field, respectively.

We call the first $d$ PDEs $W_3$-BPZ equations. They read
\begin{align} \label{W3BPZ} \tag{$W_3$-BPZ}
\mathcal D^{(m)}_{\varsigma,c} f(x_1,\cdots,x_d) = 0, \qquad m=1,\cdots,d,
\end{align}
where $ \mathcal D^{(m)}_{\varsigma,c}$ is the following third order differential operator:
\begin{align} \label{defoperatorBPZ}
 \mathcal D^{(m)}_{\varsigma,c} = \; & q_m \partial_{x_m}^3 + \frac{3(h+1)}4 \sum_{i \neq m}^d \frac{q_m \partial_{x_m} \partial_{x_i} + q_i \partial_{x_i}^2}{x_i-x_m} \\
\nonumber & + \frac{3(h+1)}{32} \sum_{i\neq m}^d \frac{[(5+h)q_m-(5h+1)q_i] \partial_{x_i} - 4[2h q_m+(h+1)q_i] \partial_{x_m}}{(x_i-x_m)^2} \\
\nonumber & - \frac{3(h+1)^2}{8} \sum_{i \neq m}^d \sum_{j \neq i,m}^d \frac{q_i \partial_{x_j}}{(x_m-x_i)(x_j-x_i)} - \frac{h(h+1)(h+5)}{16} \sum_{i \neq m}^d \frac{q_i+3q_m}{(x_i-x_m)^3} \\
\nonumber & + \frac{3h(h+1)^2}{8} \sum_{i \neq m}^d \sum_{j \neq i,m}^d  \frac{q_i}{(x_m-x_i)(x_j-x_i)^2}.
\end{align}
A formal derivation of \eqref{W3BPZ} utilizing CFT techniques developed in \cite{FL} will be presented in the Appendix. To the best of our knowledge, the operators \eqref{defoperatorBPZ} have not previously appeared in the literature, as existing works typically involve only $4$ fields, or other types of degenerate fields and/or nondegenerate fields. Nevertheless, closely related equations appeared in \cite{Rib2, CH2}.

The second property is the covariance property under Möbius transformations. More precisely, let $\varphi(x) = \frac{ax+b}{cx+d}$ with $a,b,c,d \in \mathbb R$, $ad-bc=0$ and such that $\varphi(x_1) < \varphi(x_2) < \cdots < \varphi(x_d)$. Then,
\begin{equation} \label{covariance} \tag{COV}
     f(\varphi(x_1),\cdots ,\varphi(x_d)) = \prod_{i=1}^d \varphi'(x_i)^{-h} \; f(x_1,\cdots ,x_d).
\end{equation}
It is well-known \cite{DMS} that such a covariance property is equivalent to the following three first order PDEs:
\begin{align} 
\label{globalward1} & \sum_{k=1}^d \frac{\partial f(x_1,\cdots ,x_d)}{\partial x_k} = 0, \\
\label{globalward2} & \sum_{k=1}^d \bigg[x_k \frac{\partial}{\partial x_k} + h \bigg] \; f(x_1,\cdots ,x_d) = 0, \\
\label{globalward3} & \sum_{k=1}^d \bigg[ x_k^2 \frac{\partial}{\partial x_k} + 2hx_k \bigg] \; f(x_1,\cdots ,x_d) = 0.
\end{align}
The five remaining second order PDEs are described as follows. Let $\mathcal M^{(m)}_{\varsigma,c}$ be the differential operator
\begin{align} \label{defM}
\mathcal M^{(m)}_{\varsigma,c} = \; & \sum_{i=1}^d q_i x_i^{m-1} \partial_{x_i}^2 - \frac{h+1}2 \sum_{i=1}^d \sum_{j \neq i}^d q_i x_i^{m-1} \lb \frac{h}{(x_j-x_i)^2} - \frac{\partial_{x_j}}{x_j-x_i}\rb \\
\nonumber & + \frac{(m-1)}8 (5h+1) \sum_{i=1}^d q_i x_i^{m-2} \partial_{x_i} + \frac{(m-1)(m-2)}{24} h(5h+1)\sum_{i=1}^d q_i x_i^{m-3}.
\end{align}

Then, we have
\begin{align}
    \label{ward1} \tag{WI1} & \mathcal M^{(1)}_{\varsigma,c} f(x_1,\cdots ,x_d) = 0, \\
    \label{ward2} \tag{WI2} & \mathcal M^{(2)}_{\varsigma,c} f(x_1,\cdots ,x_d) = 0, \\
    \label{ward3} \tag{WI3} & \mathcal M^{(3)}_{\varsigma,c} f(x_1,\cdots ,x_d) = 0, \\
    \label{ward4} \tag{WI4} & \mathcal M^{(4)}_{\varsigma,c} f(x_1,\cdots ,x_d) = 0, \\
    \label{ward5} \tag{WI5} & \mathcal M^{(5)}_{\varsigma,c} f(x_1,\cdots ,x_d) = 0.
\end{align}
A formal derivation of Equations \eqref{ward1}-\eqref{ward5} in the CFT framework is also presented in the Appendix. They correspond to the global Ward identities associated with the current $\mathcal W$ which encodes parts of the $W_3$-symmetry algebra of the CFT.

The last condition satisfied by $f$ is a power law bound: there exist positive constants $C$ and $p$ such that 
\begin{equation} \label{POW} \tag{POW}
f(x_1,\cdots ,x_d) \leq C \prod_{i<j}^d |x_j-x_i|^{\mu_{ij}(p)} \quad \text{with} \quad \mu_{ij}(p) := \begin{cases} -p, \quad |x_i-x_j|<1, \\ +p, \quad |x_i-x_j| \geq 1. \end{cases}
\end{equation}
We are now finally ready to define the space $\mathcal S_\varsigma^{(c)}$:
\begin{definition} \label{defSvarsigmac}
   $\mathcal S_\varsigma^{(c)} = \{f: \mathcal H_d \to \mathbb R \; | \; \text{f satisfies \eqref{W3BPZ}, \eqref{covariance}, \eqref{ward1}-\eqref{ward5}}, \eqref{POW} \}$. 
\end{definition}
We emphasize that $\mathcal S_\varsigma^{(c)}$ is a higher-rank analog of the function space associated with the Virasoro-BPZ equations defined in \cite[Definition 1]{FK1}, and studied in detail in \cite{FK1,FK2,FK3}.
\subsubsection{The space $\mathcal C_\varsigma^{(c=2)}$ of $W_3$ conformal blocks}
Let $\lambda \vdash \Summed$ be a partition, that is,  $\lambda=(\lambda_1,\lambda_2,\cdots ,\lambda_l)$ such that $\lambda_1 \geq \lambda_2 \geq\cdots \geq\lambda_l \geq 0$ and $\lambda_1+\lambda_2+\cdots +\lambda_l = \Summed$. The length of the partition $\lambda$ is then denoted by $|\lambda|=n$. A Young diagram of shape $\lambda$ is a finite collection of boxes arranged in $l$ left-justified rows with row lengths being, from top to bottom, $\lambda_1,\cdots ,\lambda_l$. Let $\text{Fill}^\lambda_\varsigma$ be the set of fillings of Young diagrams of shape $\lambda \vdash \Summed$ where the number $k$ appears $s_k$ times for all $k=1,\cdots ,\np$. We say that $\varsigma$ is the content, or weight, of a filling in $\text{Fill}^\lambda_\varsigma$. A numbering of a Young diagram is a filling such that the numbers $1,\cdots ,\Summed$ appear exactly once. A row-strict Young tableau is a filling whose entries are weakly increasing down each column and strictly increasing along each row. Similarly, a column-strict Young tableau is a filling whose numbers are weakly increasing along each row and strictly increasing down each column. In particular, column-strict Young tableaux are often called semistandard Yound tableaux in the literature. 

Let $\text{RSYT}^\lambda_\varsigma$ and $\text{CSYT}^\lambda_\varsigma$ be the set of row-strict and column-strict Young tableaux of shape $\lambda$ and content $\varsigma$, respectively. It is well-known that we have $|\text{CSYT}^\lambda_\varsigma| = |\text{RSYT}^{\bar\lambda}_\varsigma|$. A standard Young tableau is a numbering which is strictly increasing across each row and down each column. The set of numberings and of standard Young tableaux of shape $\lambda$ will be denoted $\text{NB}_\lambda$ and $\text{SYT}^\lambda$, respectively.  In particular, we have $\text{NB}_\lambda = \text{Fill}^\lambda_{(1^n)}$.

$\lambda$ and $\varsigma$ need to satisfy the following condition to guarantee that $|\text{RSYT}^\lambda_\varsigma|$ (or, equivalently, $|\text{RSYT}^\lambda_\varsigma|$) is positive. Let $\varsigma^\text{ord}$ be the composition $\varsigma$ rearranged in decreasing order, i.e. a partition. We say that two partitions $\lambda$ and $\mu$ satisfy the dominance ordering relation $\lambda \geq \mu$ if, for all $i$,
$$\lambda_1 + \cdots  + \lambda_i \geq \mu_1 + \cdots  + \mu_i$$
where we possibly extend the sequences by zeros.

\begin{definition} \label{defspechtpolynomial}
    Let $S=(i_1,\cdots ,i_r)$ be a list of positive integers. We define the Vandermonde determinant $\Delta(S)$ of the variables $x_i$ for $i \in S$:
    $$\Delta(S):= \prod_{1\leq j < k \leq r} (x_{i_j}-x_{i_k}).$$
    If $r=1$ we say $\Delta(S)=1$. 
    
    Let $F \in \text{Fill}^\lambda_\varsigma$. The Specht polynomial associated with $F$ is the polynomial
    \begin{equation} \label{spechtfactorized}
    \mathcal P_F = \prod_c \Delta(F_{.,c}),
\end{equation}
where $c$ runs through the columns of $F$ and $F_{.,c}$ denotes the entries in the $c$-th column of $N$ listed from bottom to top.
\end{definition}
For instance we have 
$$\mathcal P_{\; \begin{ytableau} 
    1 & 1 & 4 \\
    3 & 3 \\
    2
\end{ytableau}} = (x_3-x_1)^2(x_2-x_1)(x_2-x_3).$$
Let us mention that in the literature the Specht polynomials are generally defined for a numbering $N \in \text{NB}^\lambda$; here we naturally extend their definition to the case of any filling $F \in \text{Fill}^\lambda_\varsigma$. 

We are now ready to define the space $\mathcal C^{(c=2)}_\varsigma$ of $W_3$ conformal blocks. 
\begin{definition} \label{defconformalblocks}
Let $T \in {\normalfont\text{Fill}^\pi_\varsigma}$ where $\pi$ is the partition \eqref{defpi}, and let $T^t$ be the transpose of $T$, that is, $T^t$ is the reflection of $T$ through its north-west and south-east diagonal. Moreover, let $\mathcal U_T: \mathcal H_d \to \mathbb R$ be defined by
\begin{equation} \label{definitionCB}
    \mathcal U_T(x_1,\cdots,x_d) = \lb \prod_{1\leq i<j\leq n} (x_j-x_i)^{-\frac{s_i s_j}{3}} \rb \mathcal P_{T^t}(x_1,\cdots,x_d).
\end{equation}
Then, the space of $W_3$ conformal blocks $\mathcal C^{(c=2)}_\varsigma$ is defined as follows:
\begin{align} \label{defCc=2}
   \normalfont \mathcal C^{(c=2)}_\varsigma = \text{span}\lb\mathcal U_T \; | \; T \in \text{Fill}^\pi_\varsigma \rb.
\end{align}
\end{definition}

\subsection{Main results}
We now present the three classes of results of the present paper in more detail.

\subsubsection{Analytic results} The core of our analytic results is that the space of $W_3$ conformal blocks $\mathcal C^{(c=2)}_\varsigma$ in \eqref{defCc=2} is a subspace of $\mathcal S^{(c=2)}_\varsigma$. Our first result towards this goal provides a basis for the space $\mathcal C^{(c=2)}_\varsigma$:
\begin{proposition} \label{theorem1}
  A basis for $\mathcal C^{(c=2)}_\varsigma$ is given by $\left\{\mathcal U_T \; | \; T \in {\normalfont  \text{RSYT}}^\pi_\varsigma \right\}$.
\end{proposition}
The proof of Proposition \ref{theorem1} requires algebraic techniques and follows directly from Proposition \ref{basisSpechtprop}, which proves linear independence of the relevant Specht polynomials.

Then, the main result of this article is the following
\begin{theorem} \label{theorem2}
  We have $\mathcal C^{(c=2)}_\varsigma \subseteq \mathcal S^{(c=2)}_\varsigma$.
\end{theorem}
Given Proposition \ref{theorem1}, the proof of Theorem \ref{theorem2} consists of proving that each basis element satisfies \eqref{W3BPZ}, \eqref{covariance}, \eqref{ward1}-\eqref{ward5} and \eqref{POW}. While the property \eqref{POW} readily follows from Definition \ref{defconformalblocks}, the proofs of the other equations are relegated to Section \ref{sectionprooftheorem1}, because they involve rather lengthy computations. The proof of \eqref{W3BPZ} we present reveals interesting features about the $W_3$ conformal blocks and the Specht polynomials. Specifically, on the first hand, for any central charge, we prove that solutions of \eqref{W3BPZ} for any $\varsigma$ can be constructed from solutions of \eqref{W3BPZ} for $\varsigma=(1^n)$. This is an explicit realization of the fusion procedure in CFT. On the other hand, we prove that the Specht polynomials associated with numberings with three columns, which are not necessarily rectangular, satisfy a system of third order PDEs. In \cite{LPR24}, we also established with Peltola a system of second-order PDEs in the two-column case. Building on the knowledge from these two and three-column cases, we conjecture a system of $M$th-order PDEs satisfied by Specht polynomials associated with numberings with $M$ columns which are not necessarily rectangular (see Conjecture \ref{cjMcolumns}). As far as we know, this aspect of Specht polynomials has not been studied in the literature. Finally, the proofs of \eqref{covariance} and \eqref{ward1}-\eqref{ward5} are elementary and consist of straightforward, yet involved computations (see Propositions \ref{propcovariance} and \ref{propW3WardUT}). 

Proposition \ref{theorem1} implies that $\text{dim}\;\mathcal C^{(c=2)}_\varsigma$ is the number of row-strict Young tableaux of shape $\pi$ and content $\varsigma$, which is a Kostka number. Thus, Theorem \ref{theorem2} provides a lower bound for $\text{dim}\;\mathcal S^{(c=2)}_\varsigma$. In fact, this dimension is exactly what we expect from a CFT viewpoint, because the fusion rules of the $W_3$ degenerate fields coincide with the ones of the $\mathfrak{sl}_3$ fundamental representations \cite{FL}. This leads us to the following 
\begin{cj}
    $\mathcal C^{(c=2)}_\varsigma = \mathcal S^{(c=2)}_\varsigma$.
\end{cj}
Proving this conjecture requires proving an upper bound for $\operatorname{dim}{\; \mathcal S^{(c=2)}_\varsigma}$. In the lower-rank case of the Virasoro-BPZ equations studied in \cite{FK1,FK2,FK3}, an upper bound for the dimension of the solution space was proved, based on the hypoellipticity of the second order PDEs.

Apart from the proof of Theorem \ref{theorem2}, in Section \ref{sectionasy} we also provide explicit asymptotic properties for the basis elements of Proposition \ref{theorem1}. These properties imply that this basis is not the comb basis often studied in the physics literature \cite{FL}. It would be interesting to relate the two bases. 

Finally, although in Definition \ref{defconformalblocks} the Young diagram in the $W_3$-conformal blocks is assumed to be rectangular, our proof of Equations \eqref{W3BPZ} remains valid for nonrectangular Young diagrams. However, by analogy with the case of Virasoro conformal blocks \cite{KP1}, we expect that some of the eight remaining PDEs need to be modified, and some of them will not be satisfied at all. 

\subsubsection{Algebraic results}

In Section \ref{sec:algebra}, we describe several actions on the space $\mathcal C^{(c=2)}_\varsigma$. Our first result, which is a consequence of Proposition \ref{prop:repfusedsym}, uses fusion arguments to prove that $\mathcal C^{(c=2)}_\varsigma$ is an irreducible representation of the fused Hecke algebra at $q=-1$ introduced in \cite{CA}. We then show in Proposition \ref{prop:kupalgrep} that this representation descends to a certain quotient that we identify with a diagram algebra made from Kuperberg's $\mathfrak{sl}_3$ webs. Such a diagram algebra, denoted $\text{K}^\varsigma$, is a special case of spaces of webs introduced by Kuperberg in \cite{Ku}, which we refer to as Kuperberg algebra. We emphasize that the algebra $\text{K}^\varsigma$ is a higher-rank analog of the Temperley-Lieb algebra with fugacity parameter $2$.  

\subsubsection{Probabilistic results} 

We now formulate our probabilistic results which relate Kenyon and Shi's connection probabilities in the triple dimer model in terms of the $W_3$ conformal blocks of Definition \ref{defconformalblocks}. 

Let us first briefly describe the known results for the double dimer model \cite{KW1,KW2}. It is well-known that Virasoro conformal blocks of the simplest nontrivial degenerate fields lead to multiple SLE partition functions \cite{Pe19}. Moreover, the scaling limit of interfaces in the double dimer model is expected to be described by SLE$_\kappa$ at $\kappa=4$. Then, scaling limits of connection probabilities in the double dimer model are given by ratios of multiple SLE$_4$ partition functions \cite{KW1}. 

As described in the introduction, Kenyon and Shi constructed higher-rank analogs of these connection probabilities in the triple dimer model \cite{KS} where the relevant geometrical objects, replacing curves, are given by $\mathfrak sl_3$ webs. In Section \ref{connproba}, we exhibit a basis different from the one of Proposition \ref{theorem1} that we denote $\{\mathcal Z_\lambda,\; \lambda\in \Lambda^\varsigma\}$, where $\Lambda^\varsigma$ is the set of so-called reduced webs. This new basis is defined by a certain change of basis square matrix $M$, arising from seeing webs as invariant tensors. Then, an informal description of the main result of this section is the following
\begin{theorem}
    
    In the triple dimer model in the upper half plane, the scaling limit of the connection probability with topological configuration $\lambda$ and boundary conditions $T$ is given by $M_{T\lambda }\;\frac{\mathcal{Z}_\lambda}{\mathcal U_{T}}$.

\end{theorem} 
A more precise formulation of this Theorem is presented in Theorem \ref{theorem:dimer}. 

We finally remark that the formulation of connection probabilities in Theorem \ref{theorem:dimer} is directly analogous to the double dimer case \cite{KW1}. In particular, the functions $\mathcal{Z}_\lambda,\; \lambda\in \Lambda^\varsigma$ may be thought of as "pure partition functions", and the matrix $M$ is the higher-rank analog of the incidence matrix introduced in \cite{KW2}. However, whereas the entries of the incidence matrix of \cite{KW2} always equal $0$ or $1$, which indicates whether a topological configuration is possible or not given prescribed boundary conditions, the entries of our matrix $M$ can take other nonnegative values.



\section{Algebras acting on the space $\mathcal C_\varsigma^{(c=2)}$} \label{sec:algebra}
In this section, we study algebraic structures associated with the space $\mathcal C_{\varsigma}^{(c=2)}$ of Definition \ref{defconformalblocks}. 

\subsection{Irreducible modules of the Symmetric groups}

We begin by recalling facts about Specht polynomials and the representation theory of the symmetric groups. Let $\mathfrak{S}_\Summed$ be the symmetric group of $\Summed$ letters and $\mathbb C[\mathfrak{S}_\Summed]$ its group algebra. $\mathbb C[\mathfrak{S}_\Summed]$ is the algebra generated by the transpositions $\tau_i = (i,i+1) \in \mathfrak{S}_\Summed$ for $i=1,\cdots ,\Summed-1$ with relations
\begin{align}
   \nonumber &\tau_i^2=1 \quad i\in\llbracket 1, n-1\rrbracket\\
   \nonumber &\tau_i\tau_{i+1}\tau_i=\tau_{i+1}\tau_i\tau_{i+1}\quad i\in\llbracket 1, n-2\rrbracket\\
   \nonumber &\tau_i\tau_j=\tau_j\tau_i\quad \text{for } |j-i|>1.
\end{align}

The symmetric group $\mathfrak{S}_\Summed$ acts linearly on $\mathbb C[x_1,x_2,\cdots ,x_\Summed]$ by permutation of letters. A complete set of irreducible finite dimensional representations are given by so called Specht modules $V^\lambda$. They can be realized in various ways, one of them being as spaces of Specht polynomials. More precisely, it was showed in \cite[Theorem 1.1]{P75} that the space
\begin{equation} \label{defPlambda}
    P^\lambda := \text{span}\{\mathcal P_N: \; N \in \text{NB}_\lambda \} 
\end{equation}
is an irreducible $\mathfrak{S}_\Summed$-module with basis $\{\mathcal P_T \; | \; T \in \text{SYT}^\lambda\}$ and $P^\lambda\cong V^\lambda$.

\subsection{Fused symmetric groups} \label{section3p2}

In order to define various actions on $\mathcal C_{\varsigma}^{(c=2)}$, we prove some preliminary results on the Specht polynomials. Set $s_0=0$ and define
\begin{equation} \label{defqi}
    p_i = 1+ \sum_{j=0}^{i-1} s_j, \qquad i=1,\cdots ,d+1.
\end{equation}

Consider the group
\begin{align*}
    Q_\lambda = \mathfrak S_{\bar \lambda_1} \times \mathfrak S_{\bar \lambda_2} \times \cdots  \times \mathfrak S_{\bar \lambda_k}\subset \mathfrak S_n
\end{align*} where $\bar \lambda_1,\cdots ,\bar \lambda_k$ are the parts of the partition $\bar \lambda$ (in particular, $\sum_i \bar \lambda_i=n$). $Q_\lambda$ acts on $\text{Fill}^\lambda_\varsigma$ by permuting entries of a filling such that each factor $\mathfrak S_{\bar \lambda_i}$ permutes entries in the $i$th column. 
\begin{example}
    Let $T = \begin{ytableau}
        1 & 3 \\
        1 & 3 \\
        2 & 2
    \end{ytableau}$. In this case $Q_\lambda \cong \mathfrak S_3 \times \mathfrak S_3$. For instance, the element $\sigma = (13) \times \text{Id} \in Q_\lambda$ exchanges the entries lying in the first row, first column and third row, first column. Therefore,
    $\sigma.T = \begin{ytableau}
        2 & 3 \\
        1 & 3 \\
        1 & 2
    \end{ytableau}.$ Similarly, the permutation $\sigma = \text{Id} \times (12)$ leaves $T$  unchanged because it permutes two identical entries.
\end{example}
The group orbit of $F \in \text{Fill}^\lambda_\varsigma$ under the action of $Q_\lambda$ is denoted $Q_\lambda.F$. 

\begin{lemma} 
    Let $T\in {\normalfont\CSYT}$. The Specht polynomial $\mathcal P_T$ of Definition \ref{defspechtpolynomial} can be rewritten as follows:
\begin{equation} \label{spechtmonomials}
    \mathcal P_T =  \sum_{U \in Q_\lambda.T} \text{sgn}(\sigma_{T;U}) \prod_{i=1}^d x_{i}^{r^U(i)-s_i}.
\end{equation}
where $\sigma_{T;U}\in Q_\lambda$ is a permutation sending $T$ to $U$, $\text{sgn}(\sigma_{T;U})$ is its sign, and $r^U(i)$ denotes the sum of the row numbers of the entries $i$ in $U$, where we count row numbers from top to bottom.
\end{lemma}
\begin{proof}
    Each term corresponds to a summand obtained by expanding the products that are Vandermonde determinants.
\end{proof}

\begin{definition} \label{defTtilde}
Let $T\in {\normalfont\CSYT}$. To $T$ we associate a standard tableau $\tilde T$ injectively as follows. First, we relabel the entry $k$ of $T$ by $p_k$ (see \eqref{defqi}) for $k=1,\cdots ,n$. This gives a tableau $T'$. Second, we construct a word $w$ by reading the entries of $T'$ from left to right, row by row beginning with the topmost one. Third, we construct a standard tableau $\tilde T$ by relabelling the entry $l$ of $T'$ by $l+u$, where $u$ is the number of times the letter $l$ has previously appeared in $w$. 

Let $T\in {\normalfont\RSYT}$, we define $\tilde T$ in the same way except that we construct the word $w$ by reading the entries of $T'$ from top to bottom, column by column beginning with the leftmost one. 

Remark that $\Tilde{T}^t=\widetilde{T^t}$ for $T\in {\normalfont\CSYT}\cup{\normalfont\RSYT}$.
\end{definition}
\begin{definition} \label{remarkeval}
Let $f$ be a complex valued function defined on a domain $U\subset \mathbb C^n$. Suppose that $f$ can be extended by continuity on a subset $U'\subset\mathbb C^n$ such that $\{(x_1,\cdots, x_n)\in U',\;  x_{p_k}=x_{p_k+s_k-1}\text{ for each }k=1,\cdots ,d\}\neq \emptyset $. Then, by ${\normalfont [f(x_1,\cdots,x_n)]_\text{eval}}$ we mean the function of $x_1,\cdots ,x_d$ that is obtained from $f$ by the evaluations $x_{p_k}=x_{p_k+s_k-1}=y_k$ for each $k=1,\cdots ,d$ and then renaming variables as  $y_k=x_k$.
\end{definition}
\begin{definition}
   Let $P^\lambda_\varsigma$ denote the space of polynomials ${\normalfont P^\lambda_\varsigma=\{[\mathcal P]_\text{eval},\; \mathcal P\in P^\lambda\}}$. 
\end{definition}
\begin{proposition}
\label{basisSpechtprop}
    Let $\lambda$ be a partition of $n$ with three columns. The set $\{\mathcal P_T,\; T\in {\normalfont\CSYT}\}$ is a basis of $P^\lambda_\varsigma$.
\end{proposition}
This implies in particular Proposition \ref{theorem1}. To prove Proposition \ref{basisSpechtprop}, we need the following Lemma.
\begin{lemma}
\label{lemmainject}
    Let $\lambda$ be a partition of $n$ with three columns. The map which assigns to a tableau $T\in {\normalfont\CSYT}$ the sequence $(r^T(i)-s_i)_{1\leq i\leq d}$ is injective.
\end{lemma}
\begin{proof}
    Let $T,T'\in \CSYT$ such that $T\neq T'$. Let $i$ be minimum such that $T$ and $T'$ differ in the positions of $i$. If $s_i=1$ then $i$ must be in different rows of $T$ and $T'$ so that $r^T(i)\neq r^{T'}(i)$. If $s_i=2$, since the tableau has three columns, there must a box containing $i$ in the same place in $T$ and $T'$. The position of the remaining box containing $i$ differs in $T$ and $T'$ so that $r^T(i)\neq r^{T'}(i)$.
\end{proof}
\begin{proof}[Proof of Proposition \ref{basisSpechtprop}]
    Let $T\in \CSYT$, we have that $\mathcal P_T=[\mathcal P_{\tilde T}]_\text{eval}$ so that $\mathcal P_T\in P^\lambda_\varsigma$. Let $U\in \SYT$. If two entries of $U$ $i,j\in \llbracket p_k,p_{k}+s_k-1\rrbracket$ are in the same column, then $[\mathcal P_U]_\text{eval}=0$. This shows that $\{\mathcal P_T,\; T\in \CSYT\}$ spans $P^\lambda_\varsigma$. 
    
    Clearly, $(r^T(i)-s_i)_{1\leq i\leq d}$ is the unique minimum in the set $\{(r^U(i)-s_i)_{1\leq i\leq d},\; U\in Q_\lambda . T\}$ for the lexicographic order. As a consequence of \eqref{spechtmonomials}, the coefficient of the monomial $\prod_{i=1}^\Summed x_{i}^{r^T(i)-s_i}$ in $\mathcal P_T$ is $1$. This shows in particular that $P_T\neq 0$. Now, it follows from Lemma \ref{lemmainject} that the polynomials $\mathcal P_T$ for $T\in \CSYT$ are linearly independent.
\end{proof}

\subsubsection{The case $\varsigma=(1^n)$}
There exists an involutive automorphism $\omega$ of $\mathbb C[\mathfrak S_n]$ defined on permutations by
\begin{align} \label{defomega}
    \omega : \sigma \mapsto \text{sgn}(\sigma) \sigma.
\end{align}
\begin{remark}
\label{remarkomega}
Let $\Bar{\lambda}$ denote the transpose of the partition $\lambda$, and let $(V^\lambda,\rho_\lambda)$ be a Specht module. Then, $(V^\lambda,\rho_\lambda\circ \omega)$ yields a representation isomorphic to $V^{\Bar{\lambda}}$ \cite[Chapter 7]{Fu}. Let us emphasize that, under the isomorphism, the basis of polytabloids in $V^\lambda$ is not mapped to the basis of polytabloids in $V^{\Bar{\lambda}}$ but instead to the basis of so called dual polytabloids\cite[Chapter 7]{Fu}.
\end{remark}

Let $f\in \mathcal C_{(1^n)}^{(c=2)}$ such that
\begin{align}
    f=\lb \prod_{1\leq i<j\leq n} (x_j-x_i)\rb ^{-\frac{1}{3}} \mathcal P
\end{align}
where $\mathcal P\in {P}^{\Bar{\pi}}$. Moreover, consider the action of $\mathbb C [\mathfrak S_n]$ given by
\begin{align}
    x \; f=\lb \prod_{1\leq i<j\leq n} (x_j-x_i)\rb ^{-\frac{1}{3}} \omega(x)\mathcal P.
\end{align}
Thus we have $\mathcal C_{(1^n)}^{(c=2)}\cong {V}^{\pi}$ as representations of $\mathbb C [\mathfrak S_n]$. Denote by $\left\langle\tau_i,\tau_{i+1},\tau_{i+2}\right\rangle$ the subgroup of $\mathfrak S_n$ generated by $\tau_i$, $\tau_{i+1}$ and $\tau_{i+2}$. The antisymmetrizer on $4$ consecutive sites
\begin{align}
    \sum_{\sigma \in \left\langle\tau_i,\tau_{i+1},\tau_{i+2}\right\rangle} \text{sgn}(\sigma) \sigma
\end{align}
vanishes on ${V}^{\pi}$ \cite{Ma}. Therefore, this representation descends to the following quotient
\begin{align}\label{TLMquotient}
   \nonumber &\tau_i^2=1 \quad i\in\llbracket 1, n-1\rrbracket, \\
   \nonumber &\tau_i\tau_{i+1}\tau_i=\tau_{i+1}\tau_i\tau_{i+1}\quad i\in\llbracket 1, n-2\rrbracket, \\
    &\tau_i\tau_j=\tau_j\tau_i\quad \text{for } |j-i|>1, \\
    \nonumber & \sum_{\sigma \in \left\langle\tau_i,\tau_{i+1},\tau_{i+2}\right\rangle} \text{sgn}(\sigma) \sigma = 0 \quad i\in\llbracket 1, n-3\rrbracket,
\end{align}
which is called the Temperley-Lieb-Martin algebra $\text{TLM}_n$. Hence $\mathcal C_{(1^n)}^{(c=2)}$ is a representation of $\text{TLM}_n$. 

We also have the following Schur-Weyl duality type result. Denote by $V_1$ the defining representation of $\mathfrak{sl}_3$ and by $V_2=V_1\wedge V_1$ the second fundamental reprensentation.
\begin{theorem}[\cite{Ma}]
\label{thmmartin}
    The Temperley-Lieb-Martin algebra is Schur-Weyl dual to $\mathfrak{sl}_3$ on the tensor product representation $V_1^{\otimes n}$
    \begin{align}
        \textup{TLM}_n \cong \textup{End}_{\mathfrak{sl}_3}(V_1^{\otimes n})
    \end{align}
\end{theorem}

\subsubsection{The case of general $\varsigma$}

For general $\varsigma$, the role of the symmetric group is replaced by the so-called fused symmetric group defined in \cite{CA}. We first recall some useful definitions and properties. Consider the following subgroup of $\mathfrak S_n$:
\begin{equation} 
\mathfrak S_{s_1} \times \mathfrak S_{s_2} \times \cdots  \times \mathfrak S_{s_d}. 
\end{equation}
where each factor $\mathfrak S_{s_k}$ is generated by the consecutive transpositions $\tau_{p_k},\dots,\tau_{p_{k+1}-1}$.  Then consider the idempotents $\idpts$ and $\idpta$ in $\mathbb C[\mathfrak S_n]$ given by
\begin{align*}
    &\idpts := \frac{1}{s_1!\cdots  s_\np!} \prod_{k=1}^\np \sum_{\sigma \in \mathfrak S_{s_k}} \sigma, \\
    &\idpta := \frac{1}{s_1!\cdots  s_\np!} \prod_{k=1}^\np \sum_{\sigma \in \mathfrak S_{s_k}}\text{sgn}(\sigma)\;  \sigma.
\end{align*}

\begin{definition}
The algebras $\idpts \mathbb C[\mathfrak S_n] \idpts$ with unit $\idpts$ and $\idpta \mathbb C[\mathfrak S_n] \idpta$ with unit $\idpta$ were studied in \cite{CA} and are called fused symmetric group and fused Hecke algebra (at $q=-1$), respectively.
    
\end{definition}

The following Lemma is known from \cite[Appendix A.1]{CA}: 
\begin{lemma} \label{lemmaPAP} Let $A$ be a finite-dimensional semisimple algebra. Let $P$ be an idempotent element of $A$. Then the algebra $PAP$ with unit $P$ is finite-dimensional and semisimple. Moreover, if $\{R^\lambda\}_{\lambda \in \mathcal{I}}$ is a complete set of pair-wise non-isomorphic irreducible representations of A, then the set 
$$\{P(R^\lambda) \; | \; \lambda \in \mathcal{I}, \; P(R^\lambda) \neq 0\}$$
is a complete set of pair-wise non-isomorphic irreducible representations of the algebra $PAP$. 
\end{lemma}
Lemma \ref{lemmaPAP} implies that $\idpts \mathbb C[\mathfrak S_n] \idpts$ and $\idpta \mathbb C[\mathfrak S_n] \idpta$ are semisimple. The next Theorem from \cite{CA} identifies its complete set of irreducible representations.
\begin{theorem}{\cite[Theorem 6.5]{CA}} \label{proppVlambda}
The following statements hold:
\begin{enumerate}
    \item {\normalfont$\text{dim}\lb \idpts\lb V^\lambda \rb \rb = |\text{CSYT}^\lambda_\varsigma|$},
    \item A complete set of pair-wise non-isomorphic, non-zero irreducible representations of $\idpts \mathbb C[\mathfrak S_n] \idpts$ is
    $$\{\idpts\lb V^\lambda \rb\}_{\lambda \geq \varsigma^\textup{ord}}.$$ 
\end{enumerate}
\end{theorem}
We are now ready to prove that $P^{\Bar{\pi}}_\varsigma$ is a representation of $\idpts \mathbb C[\mathfrak S_n] \idpts$.
\begin{proposition}\label{prop:repfusedsym}
    The linear map 
\begin{align*}
    \psi: &\idpts(P^{\Bar{\pi}})\rightarrow P^{\Bar{\pi}}_\varsigma\\
    &\idpts \mathcal P\mapsto [\idpts \mathcal P]_\textup{eval}=[\mathcal P]_\textup{eval}
\end{align*} is an isomorphism of vector spaces. Therefore, it induces a representation of $\idpts \mathbb C[\mathfrak S_n] \idpts$ on $P^{\Bar{\pi}}_\varsigma$ as
\begin{align*}
    \idpts\, x\,\idpts\; [\mathcal P]_\textup{eval}=[\idpts\, x\,\idpts\; \mathcal P]_\textup{eval}.
\end{align*}
\end{proposition}

\begin{proof}
    From Theorem \ref{proppVlambda} we have that $\text{dim}\lb \idpts\lb P^{\Bar{\pi}} \rb \rb=\text{dim}\lb \idpts\lb V^{\Bar{\pi}} \rb \rb = |\text{CSYT}^{\Bar{\pi}}_\varsigma|$. From Proposition \ref{basisSpechtprop}, we have that $\text{dim}\lb P^{\Bar{\pi}}_\varsigma \rb = |\text{CSYT}^{\Bar{\pi}}_\varsigma|$. Moreover $\psi$ is surjective. This completes the proof.
\end{proof}

\begin{remark}$\idpta \mathbb C[\mathfrak S_n] \idpta$ is related to the fused symmetric group $\idpts \mathbb C[\mathfrak S_n] \idpts$ as follows. It is clear that $\omega(\idpts)=\idpta$ which implies that
\begin{align}
\label{isoFusedSym}
    \idpts \mathbb C[\mathfrak S_n] \idpts\cong \idpta \mathbb C[\mathfrak S_n] \idpta
\end{align}
The isomorphism and its inverse are given by
\begin{align*}
    &\idpts x \idpts \mapsto \idpta \omega(x) \idpta ,\quad x\in \mathbb C[\mathfrak S_n], \\
    &\idpta x \idpta \mapsto \idpts \omega(x) \idpts,\quad x\in \mathbb C[\mathfrak S_n].
\end{align*}
\end{remark}

It follows from Theorem \ref{proppVlambda} that a complete set of irreducible representations of $\idpta \mathbb C[\mathfrak S_n] \idpta$ is given by $\idpta(V^\lambda)$ with $\Bar{\lambda}\geq\varsigma^\text{ord}$. Let $f\in \mathcal C_{\varsigma}^{(c=2)}$ where
\begin{align}
    f=\lb \prod_{1\leq i<j\leq n} (x_j-x_i)\rb ^{-\frac{s_is_j}{3}} \mathcal P
\end{align}
with $\mathcal P\in {P}^{\Bar{\pi}}_\varsigma$. We define the action of $\idpta \mathbb C[\mathfrak S_n] \idpta$ on $\mathcal C_{\varsigma}^{(c=2)}$ by 
\begin{align}
    \idpta x \idpta f=\lb \prod_{1\leq i<j\leq n} (x_j-x_i)\rb ^{-\frac{s_is_j}{3}} \idpts \omega(x) \idpts \mathcal P,\quad x\in \mathbb C[\mathfrak S_n] .
\end{align}
\begin{proposition}
$C_{\varsigma}^{(c=2)}$ descends to the quotient $\idpta \text{TLM}_n \idpta$ and is isomorphic to $\idpta( V^{\pi})$.
\end{proposition}
\begin{proof}
    First of all, we have $P^{\Bar{\pi}}_\varsigma\cong s_\varsigma(V^{\Bar{\pi}})$ as representations of $\idpts \mathbb C[\mathfrak S_n] \idpts$. Moreover, Remark \ref{remarkomega} implies that $C_{\varsigma}^{(c=2)}\cong \idpta( V^{\pi})$ as representations of $\idpta \mathbb C[\mathfrak S_n] \idpta$. Finally, since 
    $$a_\varsigma \lb\sum_{\sigma \in \left\langle\tau_i,\tau_{i+1},\tau_{i+2}\right\rangle} \text{sgn}(\sigma) \sigma\rb a_\varsigma = 0 \quad i\in\llbracket 1, n-3\rrbracket$$ 
    on $C_{\varsigma}^{(c=2)}$, we conclude that the representation descends to the quotient.
\end{proof}

\subsection{Kuperberg algebras}
Just like for the usual Temperley-Lieb algebra, there exists a diagrammatic presentation for the Temperley-Lieb-Martin algebra. Consider the vertical strip $D=\mathbb R\times [0,1]$ with marked boundary points $(i,0)$ and $(i,1)$ for $i\in \llbracket1, d\rrbracket$. Consider oriented planar bipartite graphs embedded in $D$ such that vertices in the interior of $D$ are trivalent, and such that vertices on the boundary are univalent and given by the marked points. Moreover, we orient the graphs such that trivalent vertices are either sources or sinks and such that the edge adjacent to the marked point $(i,j)$ is oriented towards (respectively away from) the marked points if $j=0$ and $s_i=2$ or $j=1$ and $s_i=1$ (respectively $j=0$ and $s_i=1$ or $j=1$ and $s_i=2$). Such graphs, called webs, were introduced in \cite{Ku}. For instance, the following is a web for $\varsigma=(1,2,2,1)$
\begin{align*}
    \vcenter{\hbox{\includegraphics[scale=0.2]{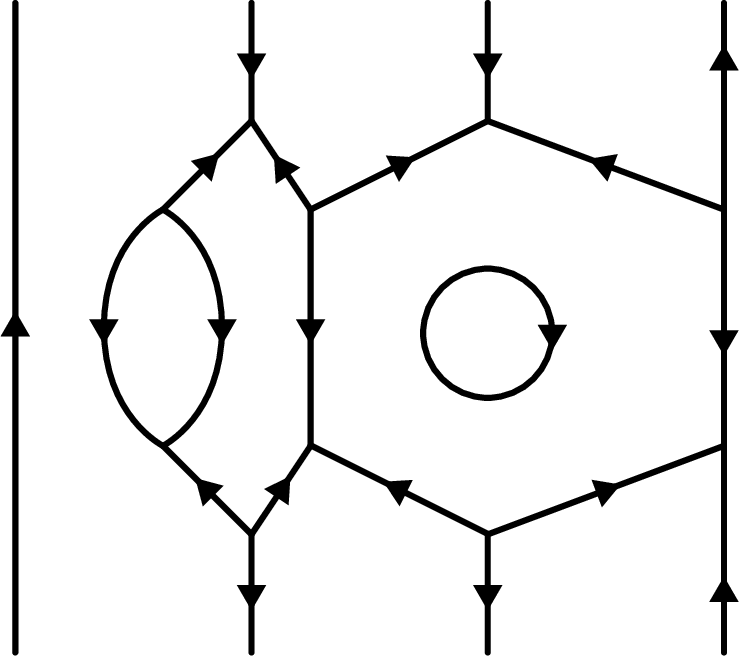}}}
\end{align*}
Denote by $\text{K}^\varsigma$ the quotient of the formal vector space over $\mathbb C$ generated by webs, considered up to isotopy and the following relations:
\begin{align}\label{eq:Kup-rels}
    \vcenter{\hbox{\includegraphics[scale=0.2]{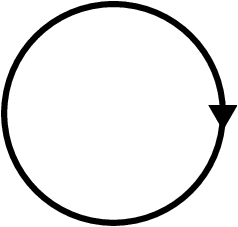}}}\quad=\quad3\;,\qquad
    \vcenter{\hbox{\includegraphics[scale=0.2]{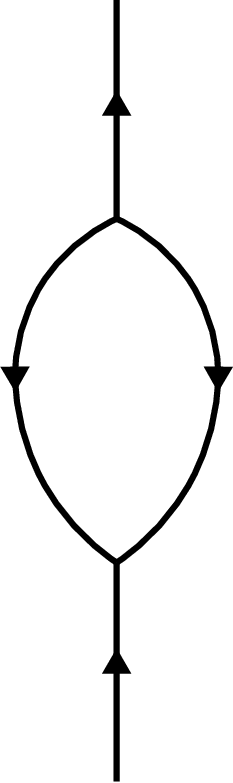}}}\quad=\quad2\;\vcenter{\hbox{\includegraphics[scale=0.2]{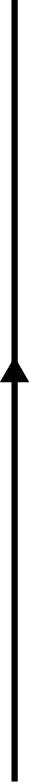}}}\;,\qquad
    \vcenter{\hbox{\includegraphics[scale=0.2]{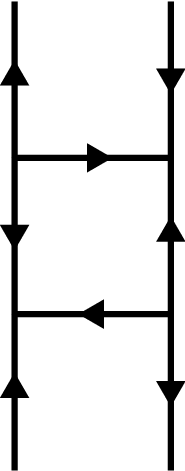}}}\quad=\quad\vcenter{\hbox{\includegraphics[scale=0.2]{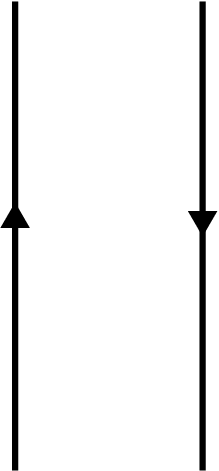}}}\quad +\quad\vcenter{\hbox{\includegraphics[scale=0.2]{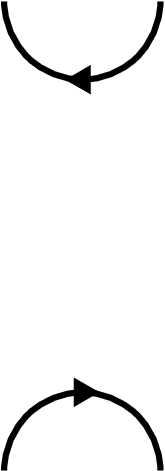}}}.
\end{align}
For instance, in $\text{K}^\varsigma$, we have the following 
\begin{align*}
    \vcenter{\hbox{\includegraphics[scale=0.2]{diagrams/exkup1.eps}}}=6\;\vcenter{\hbox{\includegraphics[scale=0.2]{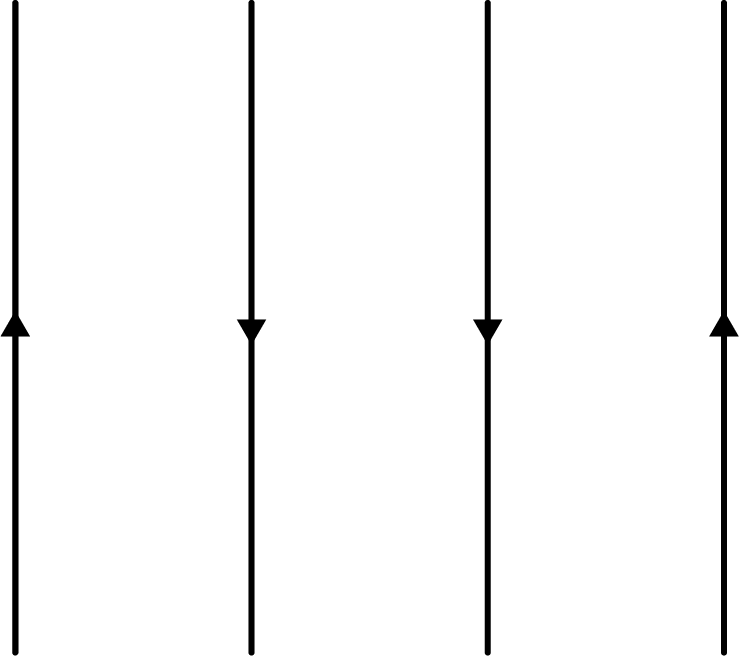}}}\quad+6\;\vcenter{\hbox{\includegraphics[scale=0.2]{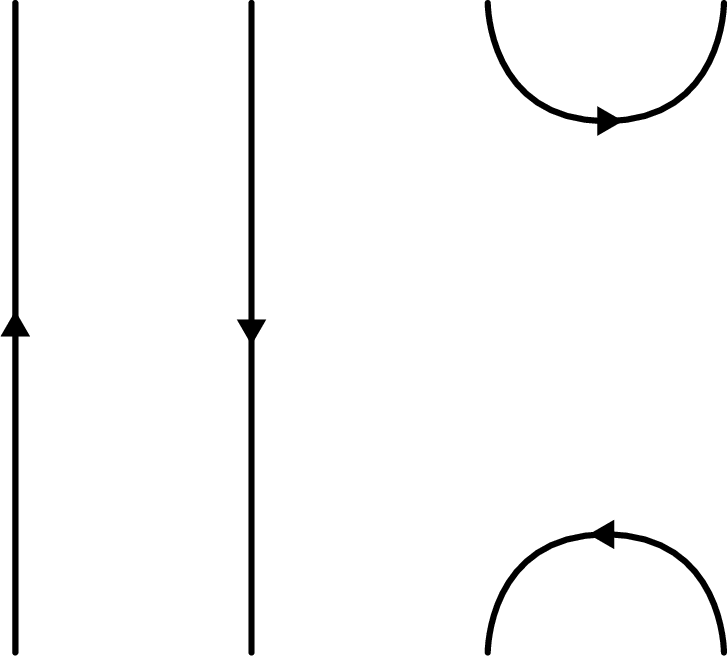}}}\quad+6\;\vcenter{\hbox{\includegraphics[scale=0.2]{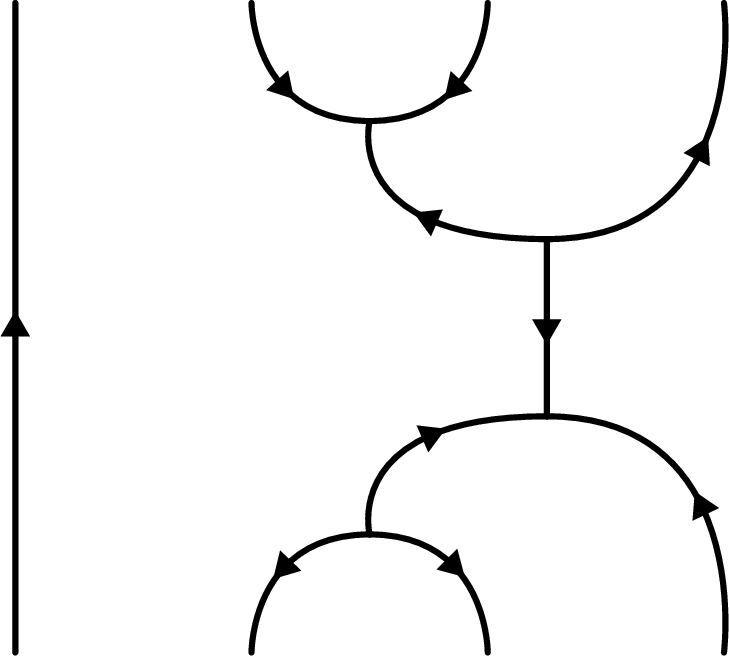}}}
\end{align*}
Given two webs $w_1$ and $w_2$, one can concatenate them placing $w_2$ on top of $w_1$ and obtain a new web $w_2w_1$ by rescaling the vertical direction. This operation can be extended bilinearly to $\text{K}^\varsigma$ to obtain an algebra that we call the Kuperberg algebra (of type $\varsigma$). The unit is given the web consisting of vertical lines joining $(i,0)$ and $(i,1)$ for $i\in \llbracket1, d\rrbracket$. 
For instance, with $\varsigma=(1,1,1)$, consider
\begin{align*}
    &w_1=\vcenter{\hbox{\includegraphics[scale=0.2]{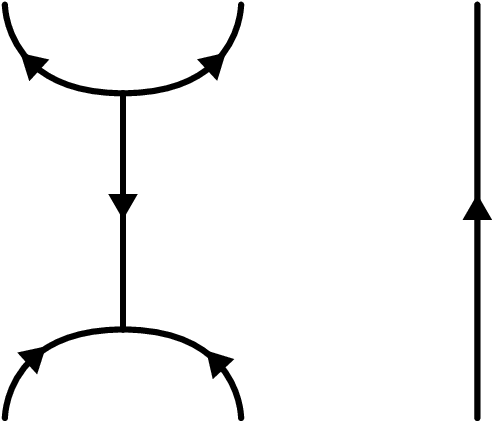}}}\qquad w_2=\vcenter{\hbox{\includegraphics[scale=0.2]{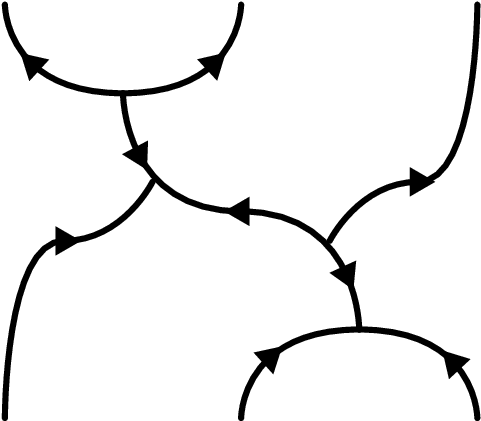}}}
\end{align*}
then we have that
\begin{align*}
    w_2w_1=\vcenter{\hbox{\includegraphics[scale=0.2]{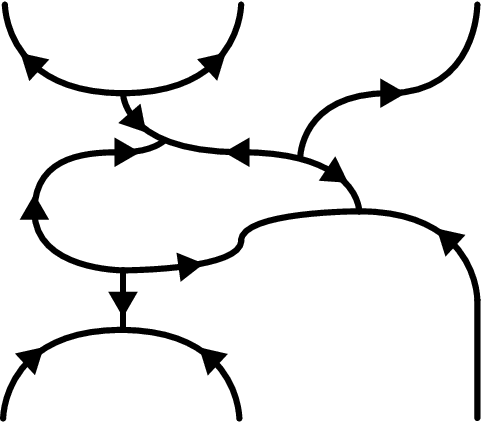}}}=\vcenter{\hbox{\includegraphics[scale=0.2]{diagrams/exkup5.eps}}}+\vcenter{\hbox{\includegraphics[scale=0.2]{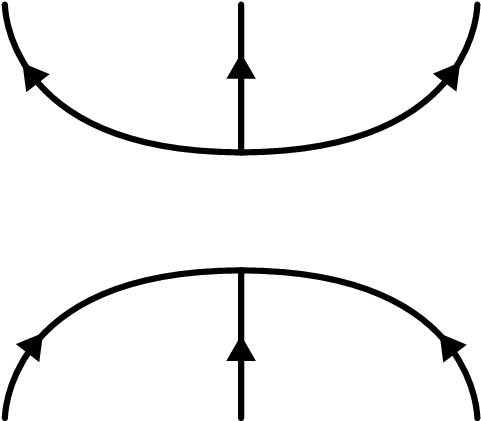}}}
\end{align*}

We have the following Schur-Weyl duality result
\begin{theorem}{\cite[Theorems 5.1 and 6.3]{Ku}}
    \begin{align*}
        \textup{K}^{\varsigma} \cong \textup{End}_{\mathfrak{sl}_3}(V_{s_1}\otimes V_{s_2}\otimes \cdots \otimes V_{s_d}).
    \end{align*}
\end{theorem}

\begin{corollary}\label{propschurweyl}
    The Temperley-Lieb-Martin algebra $\textup{TLM}_n$ is isomorphic to $\textup{K}^{(1^n)}$. The isomorphism maps a generator $\tau_i$ of $\textup{TLM}_n$ as 
    \begin{align}\label{mapkup}
    \tau_i\mapsto \quad\vcenter{\hbox{\includegraphics[scale=0.2]{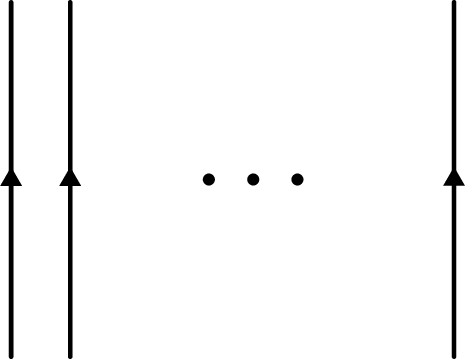}}}\quad-\quad\vcenter{\hbox{\includegraphics[scale=0.2]{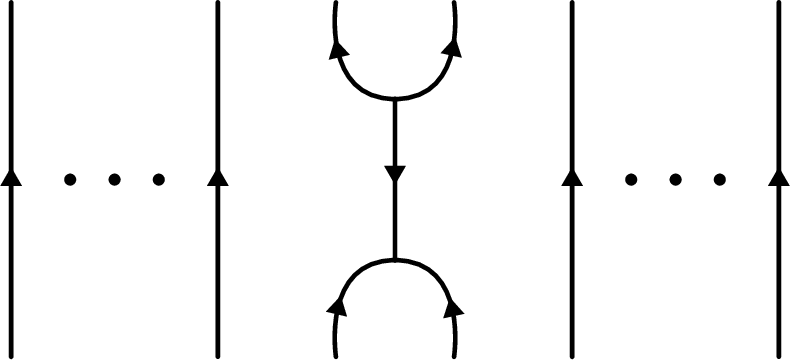}}}
\end{align}
where the connected component of the second diagram containing trivalent vertices is adjacent to boundary points $(i,0)$, $(i+1,0)$, $(i,1)$ and $(i+1,1)$.
\end{corollary}
\begin{proof}
    We have that 
    \begin{align}
        \textup{K}^{(1^n)} \cong \textup{End}_{\mathfrak{sl}_3}(V_1^{\otimes n}).
    \end{align}
    Therefore, by Theorem \ref{thmmartin},
    \begin{align}
        \textup{K}^{(1^n)} \cong \textup{TLM}_n.
    \end{align}
    It is then readily showed that the linear combination of webs \eqref{mapkup} satisfy the relations \eqref{TLMquotient}.
\end{proof}

We now describe a certain irreducible representation of $\text{K}^\varsigma$. Instead of a vertical strip, consider the lower half plane $D'=\mathbb R\times \mathbb R_{\leq 0}$ with marked points $(i,0)$ for $i\in \llbracket1, d\rrbracket$. Consider oriented planar bipartite graphs embedded in $D'$ such that vertices in the interior of $D'$ are trivalent, and such that vertices on the boundary are univalent and given by the marked points. Again, we orient the graphs such that trivalent vertices are either sources or sinks and such that the edge adjacent to the marked point $(i,0)$ is oriented towards (respectively away from) the marked points if $s_i=1$ (respectively $s_i=2$). We also call such graphs webs. The formal vector space over $\mathbb C$ generated by webs, considered up to isotopy and relations \eqref{eq:Kup-rels} is denoted by $\text{W}^\varsigma$. Given webs $x\in\text{K}^\varsigma$ and  $w\in\text{W}^\varsigma$, one obtain a new web $xw\in \text{W}^\varsigma$ by placing $x$ on top of $w$ and translating vertically. Once extended bilinearly, this endows $\text{W}^\varsigma$ with the structure of left $\text{K}^\varsigma$-module.

A basis for $\text{W}^\varsigma$ is given by so-called non-elliptic webs, or reduced webs, which are webs that do not contain any loop, digon or square. These webs cannot be reduced, meaning written as an expansion of smaller graphs, further using the relations \eqref{eq:Kup-rels}.

\begin{lemma}
    $\mathcal C_{(1^n)}^{(c=2)}$, as a representation of $\textup{TLM}_n$ is isomorphic to $\textup{W}^{(1^n)}$.
\end{lemma}
\begin{proof}
    $\mathcal C_{(1^n)}^{(c=2)}$ is isomorphic to ${V}^{\pi}$ as a representation of $ \mathbb C [\mathfrak{S}_n]$. By Lemma 4.2 of \cite{KPR}, ${V}^{\pi}$ is isomorphic to $\text{W}^{(1^n)}$. Since these representations descend to the quotient $\text{TLM}_n$, this concludes the proof.
\end{proof}

\begin{remark}
    Let $f\in \mathcal C_{(1^n)}^{(c=2)}$. Consider the braiding $b_k$ of two consecutive points $x_k$ and $x_{k+1}$ in a counterclockwise manner. It acts as follows:
\begin{align*}
   b_k f= e^{i\frac{2\pi}{3}}\lb \prod_{1\leq i<j\leq n} (x_j-x_i)\rb ^{-\frac{1}{3}} (-\tau_k) \mathcal P.
\end{align*}
Hence, on the representation $\mathcal C_{(1^n)}^{(c=2)}$, we have that
\begin{align*}
    b_k= e^{-i\frac{2\pi}{3}}\, 1 + e^{i\frac{\pi}{3}}\, e_k,
\end{align*}
and we recognize the skein relations of the $A_2$ spider \cite{Ku}.
\end{remark}

For general $\varsigma$, we have the following results

\begin{proposition}
\label{propcut}
    $\idpta \textup{TLM}_n \idpta$ is isomorphic to $\textup{K}^\varsigma$.
\end{proposition}
\begin{proof}
    $\idpta \text{TLM}_n \idpta$ is isomorphic to the diagram algebra $\idpta \text{K}^{(1^n)} \idpta$ where $\idpta$ is understood as a diagram given by vertical lines on site $p_i$ if $s_i=1$ and the following idempotent on sites $p_i$ and $p_{i}+1$ if $s_i=2$ :
    \begin{align*}
        \frac{1}{2}\vcenter{\hbox{\includegraphics[scale=0.2]{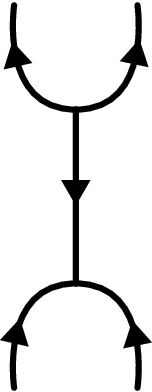}}}
    \end{align*}
    From a diagram in $\idpta \text{K}^{(1^n)} \idpta$, one can cut the above diagram as
    \begin{align*}
        \vcenter{\hbox{\includegraphics[scale=0.2]{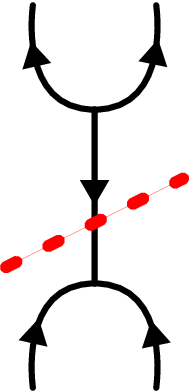}}}
    \end{align*}
    and then remove the upper part (resp. the lower part) if the idempotent sits on the top boundary (resp. the bottom boundary). We finally multiply the resulting diagram by $\sqrt{2}^{n-d}$ to account for the part of the diagram removed. This produces a diagram in $\text{K}^\varsigma$. This mapping can be extended linearly and yields an algebra homomorphism $\idpta \text{K}^{(1^n)} \idpta\rightarrow \text{K}^\varsigma$. The inverse is obtained by concatenating a diagram in $\text{K}^\varsigma$ with 
    \begin{align*}
        \frac{1}{\sqrt{2}^{3}}\vcenter{\hbox{\includegraphics[scale=0.2]{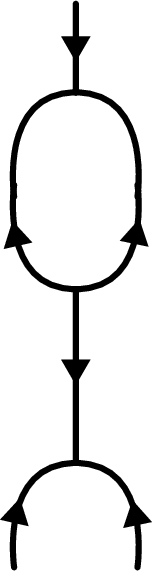}}}
    \end{align*}
    for every downward oriented arrow on the bottom boundary, and with
    \begin{align*}
        \frac{1}{\sqrt{2}^{3}}\vcenter{\hbox{\includegraphics[scale=0.2]{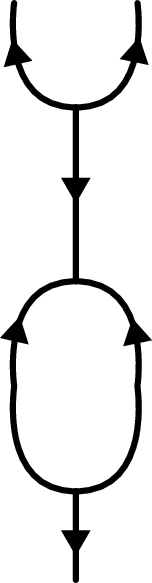}}}
    \end{align*}
    for every downward oriented arrow on the top boundary.
\end{proof}
\begin{proposition}\label{prop:kupalgrep}
    $\mathcal C_{\varsigma}^{(c=2)}$, as a representation of $\textup{K}^\varsigma$, is isomorphic to $\textup{W}^\varsigma$.
\end{proposition}
\begin{proof}
    We have that $\mathcal C_{\varsigma}^{(c=2)}\cong \idpta(\text{W}^{(1^n)})$. By analogy with Proposition \ref{propcut}, we obtain a linear isomorphism $\idpta(\text{W}^{(1^n)})\rightarrow \text{W}^\varsigma$ by removing the upper part of
    \begin{align*}
        \vcenter{\hbox{\includegraphics[scale=0.2]{diagrams/iso2.eps}}}
    \end{align*}
    appearing at the top boundary and multiplying by $\sqrt{2}^{n-d}$. The inverse is given by concatenating a diagram in $\text{W}^\varsigma$ with     
    \begin{align*}
        \frac{1}{\sqrt{2}^{3}}\vcenter{\hbox{\includegraphics[scale=0.2]{diagrams/iso4.eps}}}
    \end{align*}
    for every downward oriented arrow on the top boundary. It is clear that this linear isomorphism is an isomorphism of representations of $\text{K}^\varsigma$.
\end{proof}

$\textup{W}^\varsigma$ defines a left module for $\textup{K}^\varsigma$ and in fact can be identified with the space of vectors in $V_{s_1}\otimes V_{s_2}\otimes \cdots \otimes V_{s_d}$ invariant under the $\mathfrak{sl}_3$ action \cite{Ku}. Similarly, one can define the space $\Bar{\textup{W}}^\varsigma$ of webs obtained from $\textup{W}^\varsigma$ by horizontal reflection and reversal of the arrows. This space defines a right module for $\textup{K}^\varsigma$ and in fact can be identified with the space of $\mathfrak{sl}_3$ invariants (linear forms) of $V_{s_1}\otimes V_{s_2}\otimes \cdots \otimes V_{s_d}$. 

\section{Application to the triple dimer model } \label{sec:dimer}
This section is devoted to relating $W_3$ conformal blocks to the scaling limits of connection probabilities in the triple dimer model. The latter were first studied in \cite{KS}, see also \cite{FLL,DKS,KO,Ta1,Ta2} for recent results on $r$-fold dimer models. This section is mostly independent from Section \ref{sec:algebra}, although we need some definitions and properties of webs \cite{Ku} that were recalled there.

Let $x_1<\cdots< x_d$ be points on the real lines. To each point $x_i$ we fix an associated valence $s_i$ satisfying $\sum_i s_i = 3k$ for some $k\in \mathbb N_{>0}$. We will approximate the upper-half plane by a sequence of graphs with marked boundary vertices close to the points $\{x_i\}$. Let $\epsilon\in \mathbb R_{>0}$ and $\mathbb H_\epsilon =\epsilon \mathbb Z \times \epsilon \mathbb N_{>0}$. To use the results of \cite{Ke}, we consider simply connected Temperleyan graphs $G_N\subset \mathbb H_\epsilon$ containing the rectangle $[-N\epsilon,N\epsilon]\times[\epsilon, N\epsilon]$ such that its base point is outside of this rectangle. Let $w$ and $b$ be white and black vertices respectively on the boundary $[-N\epsilon,N\epsilon]\times \{\epsilon\}$. In the $N\to \infty$ limit, the inverse Kasteleyn matrix satisfy
\begin{align*}
    C(w,b)=\frac{2\epsilon}{\pi(b-w)} + O\left(\frac{\epsilon^2}{(b-w)^2}\right)
\end{align*}

Suppose $N$ is large enough so that $x_i\in [-N\epsilon,N\epsilon]$ for all $i$. We construct a new graph $\mathcal G_N^\epsilon$ by adding vertices to $G_N$ the following way. For $i\in \{1, \dots , d-k\}$, add a black vertex $v_i$ on the real line connected to the white vertex in $G_N$ closest to $x_i$. For $i\in S=\{d-k+1,\dots,d\}$, add a white vertex on the real line connected to black vertex in $G_N$ closest to $x_i$ and add a black vertex $v_i$ connected to this white vertex. The vertices $v_i,\;i\in \{1,\dots,d\}$ of $\mathcal G_N^\epsilon$ are called boundary vertices whereas the other vertices are called interior vertices. Note that the difference between interior white vertices and interior black vertices is $k$, called the excedance.

The graph $\mathcal G_N^\epsilon$ is adequate to use the results of \cite{FLL}\footnote{Note that the following can be rephrased in the language of \cite{KS} in a straightforward way.}. We weight all edges by 1. A web-like subgraph, or multiweb in the language of \cite{KS}, $w$ of $\mathcal G_N^\epsilon$ is a subgraph of $\mathcal G_N^\epsilon$ with edge labeled in $\{1,2,3\}$ such that the sum of multiplicities is equal to $3$ around each internal vertex and equal to $s_i$ at the boundary vertex $v_i$. A web-like subgraph $w$ defines a web in $\Bar{\textup{W}}^\varsigma$ if one forgets the edges labeled $3$ and orients the edges labeled $1$ (respectively $2$) from black to white (respectively from white to black). By abuse of notation, we denote this web by $w$ as well. Consider the following partition function on $\mathcal G_N^\epsilon$
\begin{align*}
    Z=\sum_{w\subset \mathcal G_N^\epsilon} w\,,
\end{align*}
where the sum runs over all web-like subgraphs of $\mathcal G_N^\epsilon$.

Recall that $V_1$ denotes the defining representation of $\mathfrak{sl}_3$ and $V_2=V_1\wedge V_1$ denotes the second fundamental representation. We also denote the standard basis of $V_1$ by $\{e_1,e_2, e_3\}$. For a tuple $K\in \{1,2,3\}^a$, define $e_K\in \bigwedge^a V$ by taking the wedge product of the basis vectors indexed by elements of $K$ in order. For instance, if $K=(2,1)$, $e_K=e_2\wedge e_1$.

Let $T\in \text{RSYT}^{\pi}_\varsigma$. For $i$ an entry in $T$, we denote by $K_i$ the tuple of row numbers of boxes where $i$ appears when reading the tableau from top to bottom, column by column starting from the left. For instance, if $T = \begin{ytableau}
        1 & 2 \\
        1 & 3 \\
        2 & 4
    \end{ytableau}$, we have $K_1=(1,2)$, $K_2=(3,1)$, $K_3=(2)$ and $K_4=(4)$.
Then, define the vector
\begin{align*}
    e_T= e_{K_1}\otimes \cdots \otimes e_{K_d}.
\end{align*}

Moreover, given $T$, define $C_i$, $i=1,\dots,3$, as the set of entries appearing in the $i$th row of $T$. The following result is key.
\begin{proposition}{\cite[Proposition 5.5]{FLL}}\label{propIan}
    We have\footnote{There is an additional sign in \cite[Proposition 5.5]{FLL} but it is not hard to see that our definition of $e_T$ implies that this sign is $+$.} 
\begin{align*}
    Z(e_T)=\sum_{w\subset \mathcal G_N^\epsilon} w(e_T)=Z_1 Z_2 Z_3,
\end{align*}
where $Z_i$ denote the dimer partition function on $\mathcal G_N^\epsilon\setminus \{v_j\; |\; j\notin C_i\}$.
\end{proposition}
Thus, if we consider three dimer covers, the dimer cover of color $i$ is on the graph $\mathcal G_N^\epsilon\setminus \{v_j\; |\; j\notin C_i\}$, which means that a boundary vertex $v_j$ with $j\in C_i$ is adjacent to a dimer of color $i$.

\subsection{Scaling limit}

Denote by $Z_D$ the dimer partition function on $G_N$. We want to compute the scaling limit $\lim_{\epsilon\to 0}\lim_{N\to \infty}Z(e_T)/Z_D^r$. 

\begin{proposition}\label{prop:scalinglimit}
    We have that
    \begin{align*}
        \lim_{\epsilon\to 0}\lim_{N\to \infty}Z(e_T)/Z_D^3=\mathcal{F}(x_1,\dots,x_n) \; \mathcal{U}_T(x_1,\dots,x_n) 
    \end{align*}
    where $\mathcal{U}_T$ is the conformal block associated to $T$ (see \eqref{definitionCB}) and $\mathcal{F}$ is a function independent of $T$.
\end{proposition}
\begin{proof}
Let $a=1,2,3$. From Proposition \ref{propIan}, we see that we can reduce the problem to the evaluation of $ \lim_{\epsilon\to 0}\lim_{N\to \infty}\; Z_a/Z_D$. By standard dimer arguments we have that
\begin{align}
    \lim_{\epsilon\to 0}\lim_{N\to \infty}\; Z_a/Z_D= \pm \det\left(\frac{2}{\pi(x_j-x_i)}\right)^{i\in C_a,\; i\notin S}_{j\notin C_a,\; j\in S}= \pm \left(\frac{2}{\pi}\right)^{n_a} \det\left(\frac{1}{x_j-x_i}\right)^{i\in C_a,\; i\notin S}_{j\notin C_a,\; j\in S}
\end{align}
for some sign $\pm$ and integer $n_a$ to be determined later. We can evaluate the Cauchy determinant as
\begin{align*}
    \det\left(\frac{1}{x_j-x_i}\right)^{i\in C_a,\; i\notin S}_{j\notin C_a,\; j\in S}\;&=\pm \frac{\prod\limits_{\substack{i<j\\ i,j \in C_a \text{ and } i,j\notin S\\ \text{or }i,j\notin C_a \text{ and } i,j\in S}}(x_j-x_i)}{\prod\limits_{\substack{i<j\\ i\in C_a \text{ and }i\notin S\\j\notin C_a \text{ and }j\in S}}(x_j-x_i)}\\
    &=\pm \prod\limits_{\substack{i<j\\ i,j \in C_a}}(x_j-x_i)\; \frac{\prod\limits_{\substack{i<j\\ i,j\notin C_a \text{ and } i,j\in S}}(x_j-x_i)}{\prod\limits_{\substack{i<j\\ i\in C_a \text{ and }i\notin S\\j\notin C_a \text{ and }j\in S}}(x_j-x_i)\prod\limits_{\substack{i<j\\ i,j \in C_a \\ i\in S \text{ or } j\in S}}(x_j-x_i)}\\
    &=\pm \prod\limits_{\substack{i<j\\ i,j \in C_a}}(x_j-x_i)\; F_a(x_1,\dots,x_n)
\end{align*}

Note that
\begin{align*}
    \mathcal P_{T^t}=\prod_{a=1}^3 \prod\limits_{\substack{i<j\\ i,j \in C_a}}(x_j-x_i)
\end{align*}
is the Specht polynomial associated with tableau $T^t$.

We want to show that $\prod_a F_a$ do not depend on the tableau $T$. First, observe that
\begin{align*}
    \prod\limits_{\substack{i<j\\ i\in C_a \text{ and }i\notin S\\j\notin C_a \text{ and }j\in S}}(x_j-x_i)\prod\limits_{\substack{i<j\\ i,j \in C_a \\ i\in S \text{ or } j\in S}}(x_j-x_i)&=\frac{\prod\limits_{\substack{i<j\\ i\in C_a\\ i\notin S, \; j\in S}}(x_j-x_i)}{\prod\limits_{\substack{i<j\\ i,j\in C_a\\ i\notin S, \; j\in S}}(x_j-x_i)}\prod\limits_{\substack{i<j\\ i,j \in C_a \\ i\in S \text{ or } j\in S}}(x_j-x_i)\\
    &=\prod\limits_{\substack{i<j\\ i\in C_a\\ i\notin S, \; j\in S}}(x_j-x_i)\prod\limits_{\substack{i<j\\ i,j \in C_a \\ i,j\in S }}(x_j-x_i)
\end{align*}
Then, denoting by $c_{ij}$ the number of rows in $T$ that contains both $i$ and $j$,
\begin{align*}
    \prod_a\;\prod\limits_{\substack{i<j\\ i,j\notin C_a \text{ and } i,j\in S}}(x_j-x_i)&=\prod\limits_{\substack{i<j\\ i,j\in S}}(x_j-x_i)^{3-s_i-s_j+c_{ij}}\\
    \prod_a\;\prod\limits_{\substack{i<j\\ i\in C_a\\ i\notin S, \; j\in S}}(x_j-x_i)&=\prod\limits_{\substack{i<j\\ i\notin S, \; j\in S}}(x_j-x_i)^{s_i}\\
    \prod_a\;\prod\limits_{\substack{i<j\\ i,j \in C_a \\ i, j \in S }}(x_j-x_i)&=\prod\limits_{\substack{i<j \\ i,j\in S }}(x_j-x_i)^{c_{ij}}
\end{align*}
Thus
\begin{align*}
    \prod_a F_a=\frac{\prod\limits_{\substack{i<j\\ i,j\in S}}(x_j-x_i)^{3-s_i-s_j}}{\prod\limits_{\substack{i<j\\ i\notin S, \; j\in S}}(x_j-x_i)^{s_i}}
\end{align*}
is independent of $T$.

We finally have that
\begin{align*}
    \lim_{\epsilon\to 0}\lim_{N\to \infty}Z(e_T)/Z_D^3=\pm\left(\frac{2}{\pi}\right)^{\sum_{i=1}^{d-k}s_i}\;\lb\prod_a F_a\rb\;P_{T^t}
\end{align*}
Because the LHS is positive, the sign on the RHS must be $+$. This completes the proof.
    
\end{proof}

\subsection{Connection probabilities} \label{connproba}
Consider the expansion of $Z$ with respect to the set of reduced webs $\Lambda^\varsigma$ in $\Bar{\textup{W}}^\varsigma$,
\begin{align*}
    Z&=\sum_{w\subset \mathcal G_N^\epsilon} w=\sum_{\lambda\in \Lambda^\varsigma} C_\lambda\; \lambda
\end{align*}
where $C_\lambda\geq 0$, as seen from the rules \eqref{eq:Kup-rels}. 
Let $M$ be the square matrix with rows indexed by tableaux in $\textup{RSYT}^{\pi}_\varsigma$ and columns indexed by reduced webs in $\Bar{\textup{W}}^\varsigma$ such that $M_{T\lambda }=\lambda(e_{T})$. There is a natural ordering on $\textup{RSYT}^{\pi}_\varsigma$ given by the lexicographic order if the tableau is read row by row. This order is transported to $\Bar{\textup{W}}^\varsigma$ by the bijection exhibited in \cite{KK}. It follows from \cite[Theorem 2]{KK} that, with respect to these orderings, $M$ is a unit lower triangular matrix with nonnegative integer entries. One can then define connection probabilities as the nonnegative quantities
\begin{align*}
    \textup{Pr}_\lambda^T=\frac{C_\lambda\; \lambda (e_{T})}{Z(e_{T})},\; \lambda\in \Lambda^\varsigma.
\end{align*}

\begin{theorem}\label{theorem:dimer}
    The scaling limit of connection probabilities are given by ratios of conformal blocks as 
    \begin{align*}
    \textup{P}_\lambda^T:=\lim_{\epsilon\to 0}\lim_{N\to \infty}\textup{Pr}_\lambda^T=  M_{T\lambda }\;\frac{\mathcal{Z}_\lambda}{\mathcal U_{T}}
\end{align*}
with
\begin{align} \label{relationZU}
    \mathcal{Z}_\lambda:=\sum_{U\in \textup{RSYT}^{\pi}_\varsigma} M^{-1}_{\lambda U}\;\mathcal U_{U}.
\end{align}
\end{theorem}
\begin{proof}
    
We can extract $C_\lambda$ from $Z$ using the inverse of $M$ as
\begin{align*}
    C_\lambda=\sum_{U\in \textup{RSYT}^{\pi}_\varsigma} M^{-1}_{\lambda U} Z(e_{U}).
\end{align*}
Then,
\begin{align*}
    \textup{Pr}_\lambda^T=  \frac{M_{T\lambda } \sum_{U\in \textup{RSYT}^{\pi}_\varsigma} M^{-1}_{\lambda U}Z(e_{U})}{Z(e_{T})}.
\end{align*}
The result follows from Proposition \ref{prop:scalinglimit}.

\end{proof}

\subsection{Examples}

We now provide explicit realizations of Theorem \ref{theorem:dimer}. Some of the connection probabilities were computed in \cite{KS}. 

\subsubsection{The case $\varsigma=(1,1,2,2)$}
In this case we have only two tableaux and webs which are, in order,
\begin{alignat*}{2}
    &T_1=\begin{ytableau}
        1 & 2 \\
        3 & 4 \\
        3 & 4
    \end{ytableau} \; , \qquad \qquad && \lambda_1=\vcenter{\hbox{\includegraphics[scale=0.2]{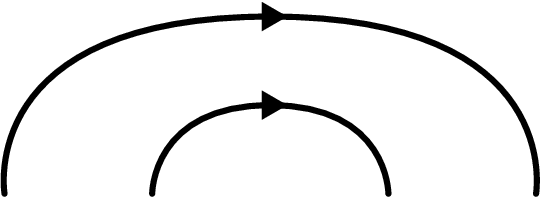}}} \; , \\[5pt]
    & T_2=\begin{ytableau}
        1 & 3 \\
        2 & 4 \\
        3 & 4
    \end{ytableau} \; , \qquad \qquad && \lambda_2=\vcenter{\hbox{\includegraphics[scale=0.2]{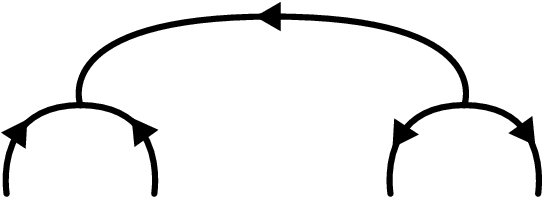}}} \; .
\end{alignat*}   The matrix $M$ is given by\begin{align*}
        M=\begin{pmatrix}
1 & 0 \\
1 & 1 
\end{pmatrix},
    \end{align*}
    and functions $\mathcal{U}_{T_1}$ and $\mathcal{U}_{T_2}$ are computed from Definition \ref{defconformalblocks}, whereas the functions $\mathcal{Z}_{T_1}$ and $\mathcal{Z}_{T_2}$ are computed from \eqref{relationZU}:
    \begin{alignat*}{2}
        & \mathcal{U}_{T_1}=\frac{(x_2-x_1)^{2/3} (x_3-x_4)^2}{(x_3-x_1)^{2/3} (x_1-x_4)^2 (x_3-x_2)^{2/3}}, \qquad&& \mathcal{Z}_{\lambda_1}=\frac{(x_2-x_1)^{2/3} (x_3-x_4)^2}{(x_3-x_1)^{2/3} (x_1-x_4)^2 (x_3-x_2)^{2/3}},\\
        &\mathcal{U}_{T_2}=\frac{(x_3-x_1)^{1/3} (x_2-x_4) (x_3-x_4)}{(x_2-x_1)^{1/3} (x_1-x_4)^2 (x_3-x_2)^{2/3}},\qquad && \mathcal{Z}_{\lambda_2}=\frac{(x_3-x_2)^{1/3} (x_3-x_4)}{(x_2-x_1)^{1/3} (x_3-x_1)^{2/3} (x_1-x_4)}.
    \end{alignat*}
    Setting $T=T_2$, this leads to
    \begin{align*}
        &\textup{P}_{\lambda_1}^T=\frac{(x_2-x_1)(x_4-x_3)}{(x_3-x_1)(x_4-x_2)}, \\
        &\textup{P}_{\lambda_2}^T=\frac{(x_3-x_2)(x_4-x_1)}{(x_3-x_1)(x_4-x_2)}.
    \end{align*}

\subsubsection{The case $\varsigma=(1,1,1,1,1,1)$} In this case the tableaux and webs are, in order, \begin{alignat*}{2}
    &T_1=\begin{ytableau}
        1 & 2 \\
        3 & 4 \\
        5 & 6
    \end{ytableau}\;, \qquad \qquad && \lambda_1=\vcenter{\hbox{\includegraphics[scale=0.2]{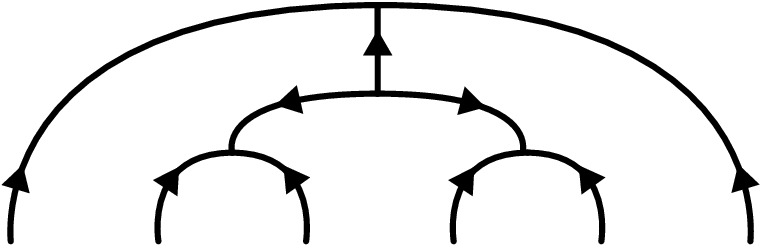}}}\;,\\[5pt]
    & T_2=\begin{ytableau}
        1 & 2 \\
        3 & 5 \\
        4 & 6
    \end{ytableau}\;, \qquad \qquad && \lambda_2=\vcenter{\hbox{\includegraphics[scale=0.2]{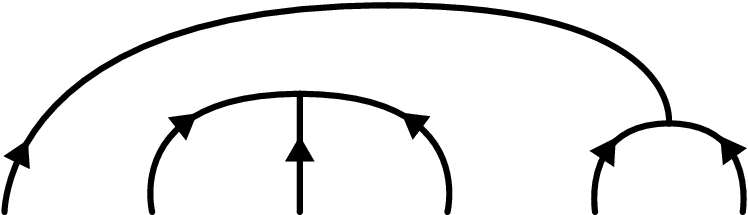}}}\;,\\[5pt]
    & T_3=\begin{ytableau}
        1 & 3 \\
        2 & 4 \\
        5 & 6
    \end{ytableau}\;, \qquad \qquad && \lambda_3=\vcenter{\hbox{\includegraphics[scale=0.2]{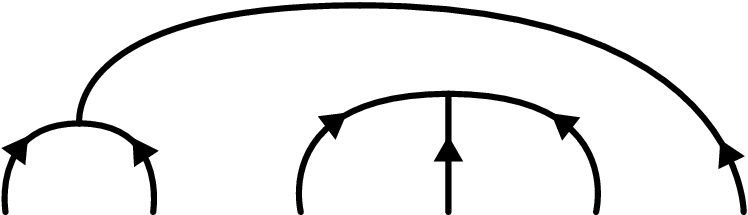}}}\;,\\[5pt]
    & T_4=\begin{ytableau}
        1 & 3 \\
        2 & 5 \\
        4 & 6
    \end{ytableau}\;, \qquad \qquad && \lambda_4=\vcenter{\hbox{\includegraphics[scale=0.2]{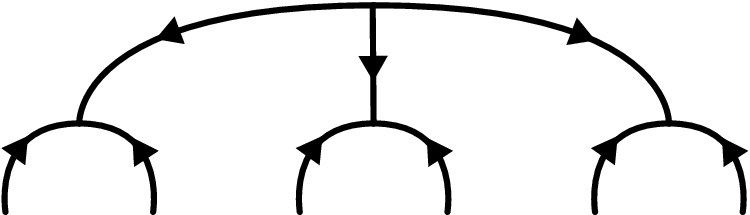}}}\;,\\[5pt]
    & T_5=\begin{ytableau}
        1 & 4 \\
        2 & 5 \\
        3 & 6
    \end{ytableau}\;, \qquad \qquad && \lambda_5=\vcenter{\hbox{\includegraphics[scale=0.2]{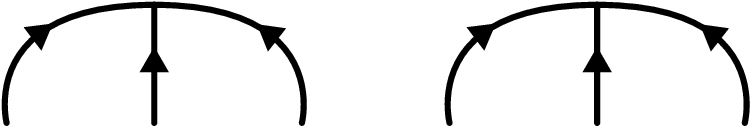}}}\;.
\end{alignat*} The matrix $M$ is given by\begin{align*}
        M=\begin{pmatrix}
1 & 0 & 0 & 0 & 0 \\
1 & 1 & 0 & 0 & 0 \\
1 & 0 & 1 & 0 & 0 \\
1 & 1 & 1 & 1 & 0 \\
1 & 1 & 1 & 1 & 1  
\end{pmatrix}.
    \end{align*} 
We have
\begin{alignat*}{2}
    & \mathcal{U}_{T_1}=\mathcal N (x_2-x_1) (x_4-x_3) (x_6-x_5),\qquad &&\mathcal{Z}_{\lambda_1}=\mathcal N(x_2-x_1) (x_4-x_3) (x_6-x_5),\\
    & \mathcal{U}_{T_2}=\mathcal N (x_2-x_1) (x_5-x_3) (x_6-x_4),\qquad &&\mathcal{Z}_{\lambda_2}=\mathcal N(x_2-x_1) (x_6-x_3) (x_5-x_4)),\\
    & \mathcal{U}_{T_3}=\mathcal N(x_3-x_1) (x_4-x_2) (x_6-x_4),\qquad &&\mathcal{Z}_{\lambda_3}=\mathcal N(x_4-x_1) (x_3-x_2) (x_6-x_5),\\
    & \mathcal{U}_{T_4}=\mathcal N(x_3-x_1) (x_5-x_2) (x_6-x_4),\qquad &&\mathcal{Z}_{\lambda_4}=\mathcal N(x_6-x_1) (x_3-x_2) (x_5-x_4),\\
    & \mathcal{U}_{T_5}=\mathcal N(x_4-x_1) (x_5-x_2) (x_6-x_3),\qquad &&\mathcal{Z}_{\lambda_5}=\mathcal N(x_6-x_1) (x_5-x_2) (x_4-x_3),
\end{alignat*}
where $\mathcal N=\lb \prod_{1\leq i<j\leq 6} (x_j-x_i)^{-\frac{s_i s_j}{3}} \rb$.
Setting $T=T_5$, this leads to
    \begin{align*}
        &\textup{P}_{\lambda_1}^T=\frac{(x_2-x_1)(x_4-x_3)(x_6-x_5)}{(x_4-x_1)(x_5-x_2)(x_6-x_3)},\\
        &\textup{P}_{\lambda_2}^T=\frac{(x_2-x_1)(x_5-x_4)}{(x_4-x_1)(x_5-x_2)},\\
        &\textup{P}_{\lambda_3}^T=\frac{(x_3-x_2)(x_6-x_5)}{(x_5-x_2)(x_6-x_3)},\\
        &\textup{P}_{\lambda_4}^T=\frac{(x_3-x_2)(x_5-x_4)(x_6-x_1)}{(x_4-x_1)(x_5-x_2)(x_6-x_3)},\\
        &\textup{P}_{\lambda_5}^T=\frac{(x_4-x_3)(x_6-x_1)}{(x_4-x_1)(x_6-x_3)}.
    \end{align*}

\subsubsection{The case $\varsigma=(1,2,1,2,1,2)$} We have
\begin{alignat*}{2}
    &T_1=\begin{ytableau}
        1 & 2 & 3\\
        2 & 4 & 6\\
        4 & 5 & 6
    \end{ytableau}\;, \qquad \qquad && \lambda_1=\vcenter{\hbox{\includegraphics[scale=0.2]{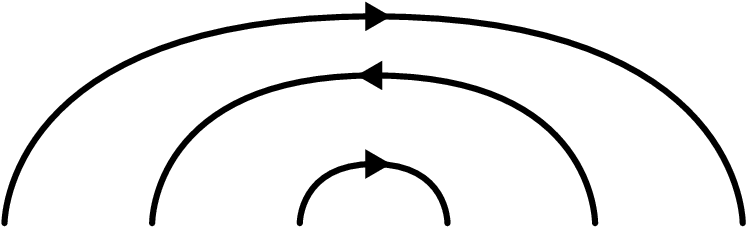}}}\;,\\[5pt]
    & T_2=\begin{ytableau}
        1 & 2 & 4\\
        2 & 3 & 6\\
        4 & 5 & 6
    \end{ytableau}\;, \qquad \qquad && \lambda_2=\vcenter{\hbox{\includegraphics[scale=0.2]{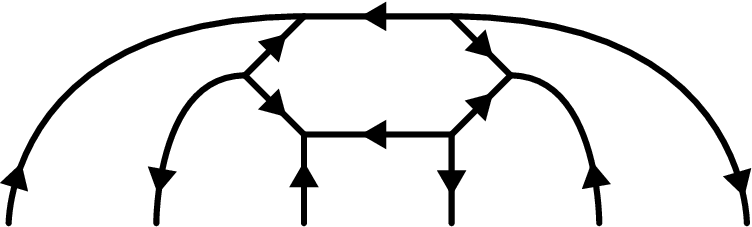}}}\;,\\[5pt]
    & T_3=\begin{ytableau}
        1 & 2 & 4\\
        2 & 4 & 6\\
        3 & 5 & 6
    \end{ytableau}\;, \qquad \qquad && \lambda_3=\vcenter{\hbox{\includegraphics[scale=0.2]{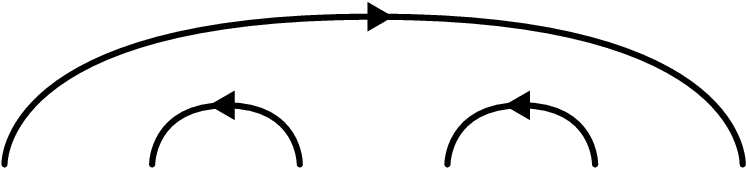}}}\;,\\[5pt]
    & T_4=\begin{ytableau}
        1 & 2 & 5\\
        2 & 4 & 6\\
        3 & 4 & 6
    \end{ytableau}\;, \qquad \qquad && \lambda_4=\vcenter{\hbox{\includegraphics[scale=0.2]{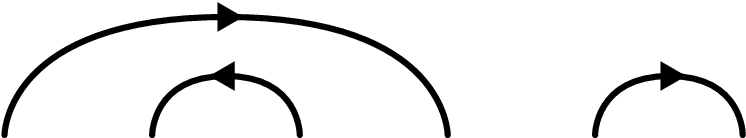}}}\;,\\[5pt]
    & T_5=\begin{ytableau}
        1 & 3 & 4\\
        2 & 4 & 6\\
        2 & 5 & 6
    \end{ytableau}\;, \qquad \qquad && \lambda_5=\vcenter{\hbox{\includegraphics[scale=0.2]{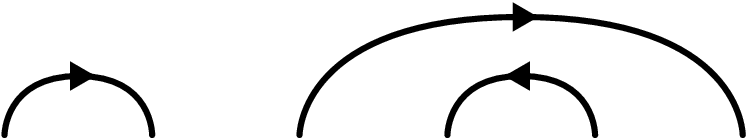}}}\;,\\[5pt]
    & T_6=\begin{ytableau}
        1 & 3 & 5\\
        2 & 4 & 6\\
        2 & 4 & 6
    \end{ytableau}\;, \qquad \qquad && \lambda_6=\vcenter{\hbox{\includegraphics[scale=0.2]{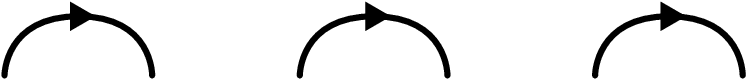}}}\;.
\end{alignat*} 
The matrix $M$ is given by
\begin{align*}
        M=\begin{pmatrix}
 1 & 0 & 0 & 0 & 0 & 0 \\
 1 & 1 & 0 & 0 & 0 & 0 \\
 1 & 1 & 1 & 0 & 0 & 0 \\
 0 & 1 & 1 & 1 & 0 & 0 \\
 0 & 1 & 1 & 0 & 1 & 0 \\
 1 & 2 & 1 & 1 & 1 & 1 \\
\end{pmatrix}.
    \end{align*}
We have
\begin{align*}
    & \mathcal{U}_{T_1}=\mathcal N(x_2-x_1) (x_3-x_1) (x_3-x_2) (x_4-x_2) (x_6-x_2) (x_5-x_4) (x_4-x_6)^2 (x_6-x_5),\\
    & \mathcal{U}_{T_2}=\mathcal N(x_2-x_1) (x_4-x_1) (x_3-x_2) (x_4-x_2) (x_6-x_2) (x_6-x_3) (x_5-x_4) (x_6-x_4) (x_6-x_5),\\
    & \mathcal{U}_{T_3}=\mathcal N(x_2-x_1) (x_4-x_1) (x_2-x_4)^2 (x_6-x_2) (x_5-x_3) (x_6-x_3) (x_6-x_4) (x_6-x_5),\\
    & \mathcal{U}_{T_4}=\mathcal N(x_1-x_2) (x_1-x_5) (x_2-x_4) (x_2-x_5) (x_2-x_6) (x_4-x_3) (x_3-x_6) (x_4-x_6)^2,\\
    & \mathcal{U}_{T_5}=\mathcal N(x_1-x_3) (x_1-x_4) (x_4-x_2) (x_2-x_5) (x_2-x_6)^2 (x_3-x_4) (x_4-x_6) (x_5-x_6),\\    
    & \mathcal{U}_{T_6}=\mathcal N(x_3-x_1) (x_5-x_1) (x_2-x_4)^2 (x_2-x_6)^2 (x_5-x_3) (x_4-x_6)^2,\\
    &\mathcal{Z}_{\lambda_1}=\mathcal N(x_2-x_1) (x_3-x_1) (x_3-x_2) (x_4-x_2) (x_6-x_2) (x_5-x_4) (x_4-x_6)^2 (x_6-x_5),\\
    &\mathcal{Z}_{\lambda_2}=\mathcal N(x_2-x_1) (x_6-x_1) (x_3-x_2) (x_4-x_2) (x_6-x_2) (x_3-x_4) (x_5-x_4) (x_4-x_6) (x_6-x_5),\\
    &\mathcal{Z}_{\lambda_3}=\mathcal N(x_1-x_2) (x_1-x_4) (x_2-x_4) (x_2-x_5) (x_2-x_6) (x_4-x_3) (x_3-x_6) (x_4-x_6) (x_5-x_6),\\
    &\mathcal{Z}_{\lambda_4}=\mathcal N(x_1-x_2) (x_1-x_6) (x_2-x_4) (x_2-x_6)^2 (x_4-x_3) (x_3-x_5) (x_4-x_5) (x_4-x_6),\\
    &\mathcal{Z}_{\lambda_5}=\mathcal N(x_1-x_5) (x_1-x_6) (x_3-x_2) (x_2-x_4)^2 (x_2-x_6) (x_3-x_4) (x_4-x_6) (x_5-x_6),\\
    &\mathcal{Z}_{\lambda_6}=\mathcal N(x_1-x_4) (x_1-x_6) (x_3-x_2) (x_2-x_4) (x_2-x_5) (x_2-x_6) (x_3-x_6) (x_4-x_5) (x_4-x_6),\\
\end{align*}
where $\mathcal N=\lb \prod_{1\leq i<j\leq 6} (x_j-x_i)^{-\frac{s_i s_j}{3}} \rb$.
Setting $T=T_6$, this leads to
    \begin{align*}
        &\textup{P}_{\lambda_1}^T=\frac{(x_2-x_1) (x_3-x_2) (x_4-x_5) (x_5-x_6)}{(x_1-x_5) (x_2-x_4) (x_2-x_6) (x_3-x_5)},\\
        &\textup{P}_{\lambda_2}^T=\frac{2 (x_1-x_2) (x_1-x_6) (x_2-x_3) (x_3-x_4) (x_4-x_5) (x_5-x_6)}{(x_1-x_3) (x_1-x_5) (x_2-x_4) (x_2-x_6) (x_3-x_5) (x_4-x_6)},\\
        &\textup{P}_{\lambda_3}^T=\frac{(x_1-x_2) (x_1-x_4) (x_2-x_5) (x_3-x_4) (x_3-x_6) (x_5-x_6)}{(x_1-x_3) (x_1-x_5) (x_2-x_4) (x_2-x_6) (x_3-x_5) (x_4-x_6)},\\
        &\textup{P}_{\lambda_4}^T=\frac{(x_1-x_2) (x_1-x_6) (x_3-x_4) (x_4-x_5)}{(x_1-x_3) (x_1-x_5) (x_2-x_4) (x_4-x_6)},\\
        &\textup{P}_{\lambda_5}^T=\frac{(x_1-x_6) (x_3-x_2) (x_3-x_4) (x_5-x_6)}{(x_1-x_3) (x_2-x_6) (x_3-x_5) (x_6-x_4)},\\
        &\textup{P}_{\lambda_6}^T=\frac{(x_1-x_4) (x_1-x_6) (x_2-x_3) (x_2-x_5) (x_3-x_6) (x_4-x_5)}{(x_1-x_3) (x_1-x_5) (x_2-x_4) (x_2-x_6) (x_3-x_5) (x_4-x_6)}.
    \end{align*}

\section{Proof of Theorem \ref{theorem2}} \label{sectionprooftheorem1}

This section is devoted to the proof of Theorem \ref{theorem2}. Since Proposition \ref{theorem1} provides a basis for $\mathcal C^{(c=2)}_\varsigma$, the proof of Theorem \ref{theorem2} consists of proving that each basis element satisfies \eqref{W3BPZ}, \eqref{covariance}, \eqref{ward1}-\eqref{ward5} and \eqref{POW}. The property \eqref{POW} naturally follows from Definition \ref{defconformalblocks}, henre it remains to verify the PDEs \eqref{W3BPZ}, \eqref{covariance}, \eqref{ward1}-\eqref{ward5}.

\subsection{Proof of \eqref{W3BPZ}}

\begin{proposition} \label{propeqW3forUT}
Let $T \in {\normalfont\RSYT}$ where $\lambda$ is a Young diagram with 3 columns which is not necessarily rectangular. Then, the $W_3$ conformal blocks $\mathcal U_T(x_1,\cdots,x_d)$ of Definition \ref{defconformalblocks} satisfy
\begin{equation} \label{eqW3BPZforUT}
    \mathcal D^{(j)}_{\varsigma,2} \; \mathcal U_T(x_1,\cdots,x_d) = 0, \qquad j=1,\cdots,d,
\end{equation}
where $\mathcal D^{(j)}_{\varsigma,c}$ is the third order partial differential operator defined in \eqref{defoperatorBPZ}.
\end{proposition}

Although the proof of Proposition \ref{propeqW3forUT} could be performed by brute force calculations, we choose another approach which reveals interesting features about the $W_3$ conformal blocks and the Specht polynomials. More precisely, for the proof of Proposition \ref{propeqW3forUT} we will need Lemmas \ref{propeqspecht3columns}, \ref{lemmaW3BPZvarsigmaequals1}, \ref{lemmafusion}, \ref{lemmasecondrepCBs} and \ref{lemmaanalytic}. Lemma \ref{propeqspecht3columns} proves that the Specht polynomials $\mathcal P_{T^t}$ where $T^t$ is a numbering with $3$ columns satisfy a system of third order PDEs. Then, Lemma \ref{lemmaW3BPZvarsigmaequals1} proves \eqref{W3BPZ} in the case where $\varsigma = (1^n)$. Finally, Lemmas \ref{lemmafusion}, \ref{lemmasecondrepCBs} and \ref{lemmaanalytic} are utilized to show that solutions of \eqref{W3BPZ} for any $\varsigma$ can be constructed solely from certain solutions of \eqref{W3BPZ} for $\varsigma = (1^n)$. 

\subsubsection{The case $\varsigma = (1^n)$.}

\begin{lemma} \label{propeqspecht3columns}
Let $N \in \textup{NB}^\lambda$ where $\lambda$ is a Young diagram with $3$ columns and $|\lambda|=n$. Then, for all $m=1,\dots,n$ we have

\begin{align} \label{equationforspechtpolynomials}
       \bigg( \partial_{x_m}^3 & + \sum_{i \neq m} \frac{\partial_{x_m}^2 + \partial_{x_i} \partial_{x_m} + \partial_{x_i}^2}{x_i-x_m} + \frac12 \sum_{i \neq m} \sum_{j \neq i,m} \frac{\partial_{x_m} + \partial_{x_i} + \partial_{x_j}}{(x_i-x_m)(x_j-x_m)} \\
    \nonumber & + \frac16 \sum_{i \neq m} \sum_{j \neq m,i} \sum_{k \neq m,i,j} \frac{1}{(x_i-x_m)(x_j-x_m)(x_k-x_m)} \bigg) \mathcal P_N(x_1,\cdots,x_n) = 0.
    \end{align}
\end{lemma}
\begin{proof}
    We first write 
    \begin{equation} \label{dePTfphi}
        \mathcal P_N(x_1,\cdots,x_n) = \prod_{1<i\leq j \leq n} (x_j-x_i)^{\phi_N(j,i)},
    \end{equation}
    where $\phi_N(j,i) = 1$ if $i,j$ lie in the same column in $N$, or $0$ otherwise. In particular, $\phi_N(j,i) = \phi_N(i,j)$ and $\phi_N(i,i) = 0$. Substitution of the derivative
    \begin{align*}
        & \partial_{x_m} \mathcal P_T(x_1,\cdots,x_n) = \lb \sum_{i \neq m} \frac{\phi_N(m,i)}{x_m-x_i} \rb \mathcal P_N(x_1,\cdots,x_n)
    \end{align*}
   into \eqref{equationforspechtpolynomials} results into a long expression involving simple, double and triple sums. Bringing all double sums to the form $\sum_{i \neq m} \sum_{j \neq m,i}$ and all triple sums to the form $\sum_{i \neq m} \sum_{j \neq m,i} \sum_{k \neq m,i,j}$, tedious but straightforward computations show that the left-hand side of \eqref{equationforspechtpolynomials} reduces to
   $$\sum_{i \neq m} \sum_{j \neq m,i} \sum_{k \neq m,i,j} \frac{d(i,j,k)}{(x_i-x_m)(x_j-x_m)(x_k-x_m)},$$
   where
   \begin{align*}
        d(i,j,k) = \; & \frac16 - \phi_N(m,i)\phi_N(m,j)\phi_N(m,k) + \phi_N(m,j)\phi_N(m,k) \\
       + & \frac12 \phi_N(m,j)\phi_N(i,k) + \frac13 \phi_N(i,j)\phi_N(i,k) - \frac12 \phi_N(i,k) - \frac12 \phi_N(m,k).
   \end{align*}
   By symmetry of the summations with respect to $i,j,k$, we have
   \begin{align*}
    &   \sum_{i \neq m} \sum_{j \neq m,i} \sum_{k \neq m,i,j} \frac{d(i,j,k)}{(x_i-x_m)(x_j-x_m)(x_k-x_m)} = \sum_{i \neq m} \sum_{j \neq m,i} \sum_{k \neq m,i,j} \frac{g(i,j,k)}{(x_i-x_m)(x_j-x_m)(x_k-x_m)},
   \end{align*}
   where 
   \begin{align*}
       g(i,j,k) = \frac16 \lb d(i,j,k)+d(j,i,k)+d(i,k,j)+d(k,j,i)+d(j,i,k)+d(k,i,j) \rb,
   \end{align*}
   or, more explicitly, 
\begin{align*}
    g(i,j,k) = \; & - \phi_N(m,i)\phi_N(m,j)\phi_N(m,k) + \frac16 \\
    &- \frac16 \lb \phi_N(m,i) + \phi_N(m,j) + \phi_N(m,k) + \phi_N(i,j) + \phi_N(i,k) + \phi_N(j,k) \rb \\
    & + \frac16 \lb \phi_N(m,i)\phi_N(j,k) + \phi_N(m,j)\phi_N(i,k) + \phi_N(m,k)\phi_N(i,j)\rb \\
    & + \frac13 \lb \phi_N(m,i)\phi_N(m,j) + \phi_N(m,i)\phi_N(m,k) + \phi_N(m,j)\phi_N(m,k)\rb \\
    & + \frac19 \lb \phi_N(i,j)\phi_N(i,k) + \phi_N(j,i) \phi_N(j,k) + \phi_N(k,i)\phi_N(k,j) \rb.
\end{align*}
   The last step of the proof is to show that $g(i,j,k) = 0$. We proceed case by case. There are, up to permutations of $i,j,k$, six different cases to consider. In each case, we give all possible values for $\phi_N$ and the identity $g(i,j,k)=0$ readily follows.
   
   \textbf{First case}: $i,j,k$ in three different columns. We have $\phi_N(i,j)=\phi_N(j,k)=\phi_N(i,k)=0$. Without loss of generality we can suppose that $i$ lies in the same column as $m$, which gives $\phi_N(m,i)=1,\;\phi_N(m,j)=\phi_N(m,k)=0$.
   
   \textbf{Second case}: $i,j,k$ in the same column different from the one of $m$. We have $\phi_N(i,j)=\phi_N(j,k)=\phi_N(i,k)=1$ and $\phi_N(m,i)=\phi_N(m,j)=\phi_N(m,k)=0$.

   \textbf{Third case}: $i,j,k,m$ in the same column. We have $\phi_N(i,j)=\phi_N(j,k)=\phi_N(i,k)=\phi_N(m,i)=\phi_N(m,j)=\phi_N(m,k)=1$.

   \textbf{Fourth case}: $i,j,m$ in the same column, $k$ in a different column. We have $\phi_N(i,j)=\phi_N(m,i)=\phi_N(m,j)=1$ and $\phi_N(j,k)=\phi_N(i,k)=\phi_N(m,k)=0$. 

   \textbf{Fifth case}: $i,m$ in the same column, $j,k$ in the same column different from the one of $i,m$. We have $\phi_N(j,k)=\phi_N(m,i)=1$ and $\phi_N(i,j)=\phi_N(i,k)=\phi_N(m,k)=\phi_N(m,j)=0$. 

   \textbf{Sixth case}: $m,i,k$ in different columns, $j,i$ in the same column. We have $\phi_N(i,j)=1$ and $\phi_N(j,k)=\phi_N(i,k)=\phi_N(m,k)=\phi_N(m,j)=\phi_N(m,i)=0$. 
\end{proof}
\begin{remark} \label{remarkspecht2columns}
Let $N \in \textup{NB}^\lambda$ where $\lambda$ is a Young diagram with $2$ columns and $|\lambda|=n$. Then, it follows from \cite{LPR24} that for all $m=1,\dots,n$ we have
    \begin{equation}
       \lb \partial_{x_m}^2 + \sum_{i \neq m} \frac{\partial_{x_i} + \partial_{x_m}}{x_i-x_m} + \frac12 \sum_{i \neq m} \sum_{j \neq i,m} \frac{1}{(x_i-x_m)(x_j-x_m)} \rb \mathcal P_N(x_1,\dots,x_n) = 0.
        \end{equation}
\end{remark}
Combining Lemma \ref{propeqspecht3columns} and Remark \ref{remarkspecht2columns}, we then formulate the following conjecture:
\begin{cj} \label{cjMcolumns}
Let $S_k^{(m)} = \{(i_1,\dots,i_k) \in \llbracket 1,n \rrbracket \; | \; i_l \neq m \; \text{for all} \; l=1,\dots,k \; \text{and} \; i_l \neq i_p \; \text{for} \; l \neq p \; \text{and for all} \; l,p =1,\dots,k\}$. Moreover, let $N \in \textup{NB}^\lambda$ where $\lambda$ is a Young diagram with $M$ columns and $|\lambda|=n$. Then, for all $m=1,\dots,n$ we have
\begin{equation}
    \lb \partial_{x_m}^M + \sum_{k=1}^M \frac{1}{k!} \sum_{(i_1,\dots,i_k) \in S_k^{(m)}} \left[\prod_{l=1}^k \frac{1}{x_{i_l}-x_m}\right] \sum_{\substack{(\alpha_0,\dots,\alpha_k) \in \mathbb Z_{\geq 0}^{k+1} \\ \alpha_0 + \cdots + \alpha_k = M-k}} \partial_{x_m}^{\alpha_0} \prod_{l \geq 1}^k \partial^{\alpha_l}_{x_{i_l}} \rb \mathcal P_N(x_1,\dots,x_n) = 0.
\end{equation}
\end{cj}

Thanks to Lemma \ref{propeqspecht3columns} we are now ready to prove \eqref{W3BPZ} in the case where $\varsigma = (1^n)$. 

\begin{lemma} \label{lemmaW3BPZvarsigmaequals1}
    Let $N \in \textup{NB}^{\lambda}$ where $\lambda$ is a Young diagram with 3 rows which is not necessarily rectangular, and let $\mathcal P_{N^t}(x_1,\cdots,x_n)$ be its associated Specht polynomial. Then, the conformal block functions $\mathcal U_N(x_1,\dots,x_n)$ of Definition \ref{defconformalblocks} satisfy
    \begin{equation} 
        \mathcal D^{(m)}_{(1^n),2} \mathcal U_N(x_1,\cdots,x_n) = 0.
    \end{equation}
\end{lemma}
\begin{proof}
From Definition \ref{defconformalblocks} we have
\begin{equation} \label{Dvarsigmaequals1}
\mathcal D^{(m)}_{(1^n),2} \mathcal U_N(x_1,\cdots,x_n) = \mathcal D^{(m)}_{(1^n),2} \lb \lb \prod_{1<i\leq j \leq n} (x_j-x_i)^{-\frac13} \rb \mathcal P_{N^t} (x_1,\cdots,x_n) \rb.
\end{equation}
We now express these equations in terms of $\mathcal P_{N^t}$. Substitution of the derivative 
$$\partial_{x_m} \lb \prod_{1<i\leq j \leq n} (x_j-x_i)^{-\frac13} \rb = \frac13 \prod_{1<i\leq j \leq n} (x_j-x_i)^{-\frac13} \sum_{i \neq m} \frac{1}{x_i-x_m}$$
yields a long expression involving simple, double and triple sums. Rewriting all double sums in the form $\sum_{i \neq m} \sum_{j \neq i,m}$ and all triple sums in the form $\sum_{i \neq m} \sum_{j \neq i,m} \sum_{k \neq i,m,j}$, tedious but straightforward computations show that \eqref{Dvarsigmaequals1} is equivalent to the left-hand side of \eqref{equationforspechtpolynomials} with $N \to N^t$. This equals 0 by Proposition \ref{propeqspecht3columns}. This completes the proof.
\end{proof}

\subsubsection{Case of general $\varsigma$}
The remainder of this section is devoted to the proof of \eqref{W3BPZ} for any $\varsigma$.
\begin{lemma} \label{lemmafusion}
    Let $\beta \neq 0$ so that $h+1 \neq 0$. Moreover, let $\varsigma = (s_1,\cdots,s_{j-1},1,1,s_{j+2},\cdots,s_d)$ for some $j \in \llbracket 1,d-1 \rrbracket$, and let $f:\mathcal H_n \to \mathbb R$ be a function which satisfies
    $$\mathcal D^{(i)}_{\varsigma,c} \; f(x_1,\cdots,x_d) = 0, \qquad i=1,\cdots,d,$$
    and which has the Fröbenius series expansion 
    \begin{equation} \label{frobeniusansatz}
        f(x_1,\cdots,x_d) = (x_{j+1}-x_j)^{-h} \sum_{i \geq 0} f_i(x_1,\cdots,x_j,x_{j+2},\cdots,x_d) (x_{j+1}-x_j)^i.
    \end{equation}
    Suppose that the coefficients $f_i$ are smooth functions of $x_1,\cdots,x_j,x_{j+2},\cdots,x_d$, and that for any compact subset $K$ of their domain of definition and for any multi-index $\alpha$ there exist positive constants $r_{K,\alpha}, C_{K,\alpha}$ such that the bound on partial derivatives 
\begin{equation} \label{uniformbound}
    \left|\partial^\alpha f_m(x_1,\cdots,x_j,x_{j+2},\cdots,x_d) \right| \leq C_{K,\alpha} r^{-m}_{K,\alpha}
\end{equation}
  holds for all $m \in \mathbb N$ and $(x_1,\cdots,x_j,x_{j+2},\cdots,x_n) \in K$. Then, the function $f_0$ defined by the limit
  \begin{equation}
      f_0(x_1,\cdots,x_j,x_{j+2},\cdots,x_d) = \lim_{x_{j+1} \to x_j} (x_{j+1}-x_j)^h f(x_1,\cdots,x_d)
  \end{equation}
  satisfies
    \begin{align} 
          \label{fusedBPZequations} \mathcal D_{\varsigma',c}^{(i)} f_0(x_1,\cdots,x_j,x_{j+2},\cdots,x_d) = 0, \qquad i=1, \cdots, j-1,j, j+2,\cdots,d,
    \end{align}
    where $\varsigma' = (s_1,\cdots,s_{j-1},2,s_{j+2},\cdots,s_d)$.
\end{lemma}
\begin{proof}
Throughout the proof it will be convenient to perform the change of variables $(x_j,x_{j+1}) \to (\hat x_j, \epsilon)$ such that
\begin{equation} \label{defhatxjandepesilon}
    \hat x_j = \frac{x_j+x_{j+1}}{2}, \qquad \epsilon = x_{j+1}-x_j,
\end{equation}
and then send $\hat x_j \to x_j$ at the end. By the chain rule of derivatives, we have
\begin{equation}
    \partial_{x_j} = - \partial_\epsilon + \frac12 \partial_{\hat x_j}, \qquad \partial_{x_j+1} = \partial_\epsilon + \frac12 \partial_{\hat x_j}.
\end{equation}
We will first treat the PDEs $\mathcal D_{\varsigma,c}^{(i)}f=0$ for $i=j$ and $i=j+1$ which will lead to \eqref{fusedBPZequations} for $i=j$. Then, we will study the PDEs $\mathcal D_{\varsigma,c}^{(i)}f=0$ for $i\neq j,j+1$ which will lead to \eqref{fusedBPZequations} for $i\neq j$. In each case, the strategy is to apply the differential operator $\mathcal D_{\varsigma,c}^{(i)}f=0$ to the Fröbenius series expansion \eqref{frobeniusansatz}. The assumption \eqref{uniformbound} provides a locally uniform bound on the coefficients $f_m$ and their partial derivatives, hence we exchange the order of the sum and the partial derivatives and we obtain an expression of the form 
$$\mathcal D_{\varsigma,c}^{(i)}f = \epsilon^{-h-3}\sum_{k \geq 0} c_k^{(i)} \; \epsilon^k = 0,$$
where the $c_k^{(i)}$'s are linear combinations of the $f_k$'s. The constraints $c_k^{(i)}=0$ for $k \geq 0$ lead to linear relations between the $f_k$'s. Exploiting such relations up to a certain order leads to the desired equations satisfied by $f_0$. However, since this requires tedious but straightforward computations, we will skip several computational steps in the proof. \\

\noindent \textbf{Proof of \eqref{fusedBPZequations} for $i=j$ utilizing $\mathcal D^{(j)}_{\varsigma,c} f = 0$ and $\mathcal D_{\varsigma,c}^{(j+1)} f = 0$.} We first consider the PDE $\mathcal D^{(j)}_{\varsigma,c}f=0$. The first equation $c_0^{(j)}=0$, also called indicial equation, is automatically satisfied:
\begin{equation} \label{indicialeq}
    c_0^{(j)} = h(h+1)(h+2) - \frac{3(h+1)}4 - \frac{h(h+1)(h+5)}{4} = 0.
\end{equation}
The second equation is
\begin{equation} \label{eqc1j}
    c_1^{(j)} = (h+1)(3h-1) f_1 = 0 \implies (3h-1) f_1 = 0.
\end{equation}
Next, the third equation is
\begin{align*}
    c_2^{(j)} = \; & -\frac32(h-1)(3h+1) f_2 + \frac34(1-3h) f_1 + \frac{3}{8} (3 h+1) \partial_{\hat x_j}^2 f_0 \\
    & - \frac{3}{8}h (h+1) (3 h+1) \sum_{i \neq j,j+1} \frac{f_0}{(x_i-\hat x_j)^2} + \frac{3}{8} (h+1) (3 h+1) \sum_{i \neq j,j+1} \frac{\partial_{x_i} f_0}{x_i-\hat x_j}.
\end{align*}
By \eqref{eqc1j} the term in $f_1$ vanishes, hence we obtain a linear relation between $f_2$ and $f_0$:
\begin{equation} \label{relationf2f0}
    f_2 = - \frac{h(1+h)}{4(h-1)} \sum_{i \neq j,j+1} \frac{f_0}{(x_i-\hat x_j)^2} + \frac{1+h}{4(h-1)} \sum_{i \neq j,j+1} \frac{\partial_{x_i} f_0}{x_i-\hat x_j} + \frac{1}{4(h-1)} \partial_{\hat x_j}^2 f_0.
\end{equation}
Obtaining the constraint $c_3^{(j)}=0$ involves more computations and requires both equations $\mathcal D^{(j)}_{\varsigma,c} f=0$ and $\mathcal D_{s_{j+1}}^{(j)} f=0$. More precisely, utilizing 
\begin{equation} \label{expansiondenom}
    \frac{1}{x_i - \hat x_j + \frac{\epsilon}2} = \frac{1}{x_i-\hat x_j} \sum_{l=0}^\infty \frac{(-1)^l\epsilon^l}{2^l}(x_i-\hat x_j)^{-l}, \qquad \frac{1}{x_i - \hat x_j - \frac{\epsilon}2} = \frac{1}{x_i-\hat x_j} \sum_{l=0}^\infty \frac{\epsilon^l}{2^l}(x_i-\hat x_j)^{-l},
\end{equation}
long but straightforward computations show that the equations $\mathcal D^{(j)}_{\varsigma,c} f=0$ and $\mathcal D^{(j+1)}_{\varsigma,c} f=0$ respectively yield the following constraints:
\begin{align}
  \label{c3j} & c_3^{(j)} = -\frac{3(9h(h-2)+5)}{4} f_3 + \mathcal A f_1 -\frac{9(h-1)}2 \partial_{\hat x_j} f_2 + \mathcal B f_0 = 0, \\
  \label{c3jp1} & c_3^{(j+1)} = -\lb -\frac{3(9h(h-2)+5)}{4} f_3 + \mathcal A f_1 \rb -\frac{9(h-1)}2 \partial_{\hat x_j} f_2 + \mathcal B f_0 = 0,
\end{align}
where 
\begin{align} \label{formulaB}
    \mathcal B = \; & \frac18 \partial_{\hat x_j}^3 + \frac{3(h+1)}{4} \sum_{i \neq j,j+1} \frac{1/2 \partial_{\hat x_j} \partial_{x_i} + q_i \partial_{x_i}^2}{x_i - \hat x_j} \\
    \nonumber & + \frac{3(h+1)}{32} \sum_{i \neq j,j+1} \frac{[7-h-(5h+1)q_i] \partial_{x_i} - 4[h q_m+(h+1)q_i] \partial_{x_m}}{(x_i-x_m)^2} \\
  \nonumber  & - \frac{3(h+1)^2}{8} \sum_{i \neq j,j+1} \sum_{k \neq i,j,j+1} \frac{q_i \partial_{x_k}}{(\hat x_j-x_i)(x_k-x_i)} - \frac{h(h+1)(h+5)}{16} \sum_{i \neq j,j+1} \frac{q_i}{(x_i-\hat x_j)^3} \\
    \nonumber  & + \frac{3h(h+1)(h-7)}{16} \sum_{i \neq j,j+1} \frac{1}{(x_i- \hat x_j)^3} + \frac{3h(h+1)^2}{8} \sum_{i \neq j,j+1} \sum_{j \neq i,j,j+1}  \frac{q_i}{(\hat x_j-x_i)(x_k-x_i)^2},
\end{align}
and where $\mathcal A$ is an intricate differential operator. The explicit expression for $\mathcal A$ needs not be written. In fact, the sum of \eqref{c3j} and \eqref{c3jp1} leads to the constraint  
\begin{equation} \label{relationbisf2f0}
    -\frac{9(h-1)}2 \partial_{\hat x_j} f_2 + \mathcal B f_0 = 0.
\end{equation}
Moreover, the derivative of \eqref{relationf2f0} with respect to $\hat x_j$ gives
\begin{align} \label{eqforderivativeoff2}
     -\frac{9(h-1)}2 \partial_{\hat x_j} f_2 = \; & \frac{9h(1+h)}{4} \sum_{i \neq j,j+1} \frac{f_0}{(x_i-\hat x_j)^3} - \frac{9(h+1)}{8} \sum_{i \neq j,j+1} \frac{\partial_{x_i} f_0}{(x_i-\hat x_j)^2} - \frac{9}{8} \partial_{\hat x_j}^3 f_0 \\
   \nonumber  & + \frac{9h(h+1)}{8} \sum_{i \neq j,j+1} \frac{\partial_{\hat x_j} f_0}{(x_i-\hat x_j)^2} - \frac{9(h+1)}{8} \sum_{i \neq j,j+1} \frac{\partial_{\hat x_j} \partial_{x_i} f_0}{x_i - \hat x_j}.
\end{align}
Substituting \eqref{eqforderivativeoff2} into \eqref{relationbisf2f0} and simplifying the result finally leads to
\begin{align*}
& \bigg[ - \partial_{\hat x_j}^3 + \frac{3(h+1)}4 \sum_{l \neq j,j+1} \frac{- \partial_{\hat x_j} \partial_{x_l} + q_l \partial_{x_l}^2}{x_l-\hat x_j} \\
& + \frac{3(h+1)}{32} \sum_{l\neq j,j+1} \frac{[-(5+h)-(5h+1)q_l] \partial_{x_l} - 4[-2h+(h+1)q_l] \partial_{\hat x_j}}{(x_l-\hat x_j)^2} \\
& - \frac{3(h+1)^2}{8} \sum_{l \neq j,j+1} \sum_{k \neq l,j,j+1} \frac{q_l \partial_{x_k}}{(\hat x_j-x_l)(x_k-x_l)} - \frac{h(h+1)(h+5)}{16} \sum_{l \neq j,j+1} \frac{q_l-3}{(x_l-\hat x_j)^3} \\
& + \frac{3h(h+1)^2}{8} \sum_{l \neq j,j+1} \sum_{k \neq l,j,j+1}  \frac{q_l}{(\hat x_j-x_l)(x_k-x_l)^2} \bigg] f_0 = 0,
\end{align*}
which is exactly \eqref{fusedBPZequations} for $i=j$. \\

\noindent \textbf{Proof of \eqref{fusedBPZequations} for $i\neq j$ utilizing $\mathcal D_{\varsigma,c}^{(i)} f = 0$ for $i \neq j,j+1$.} In this case it is straightforward to observe that $c_0^{(i)}$ and $c_1^{(i)}$ are identically zero. Next, we obtain $c_2^{(i)} = c_1^{(j)} = 0$ where $c_1^{(j)}$ is given in \eqref{eqc1j}. The condition $c_3^{(i)} = 0$ directly leads to \eqref{fusedBPZequations} for $i\neq j$, however it requires a substantial amount of straightforward computations. Therefore, most of the steps are not made explicit. The strategy is to first exclude all terms containing $j$ and $j+1$ in the various sums appearing in the differential operator $\mathcal D_{\varsigma,c}^{(i)}$, then apply the change of variables \eqref{defhatxjandepesilon}, and finally extract all terms proportional to $\epsilon^{-h}$ utilizing \eqref{expansiondenom}. This leads to an equation of the form 
\begin{align} \label{eqc2m} 
& \mathcal C f_0 - \frac{9(h+1)(h-1)f_2}{2(\hat x_j-x_i)} + q_i \partial_{x_i}^3 + \frac{3(h+1)}4 \sum_{l \neq i,j,j+1} \frac{q_i \partial_{x_i} \partial_{x_l} + q_l \partial_{x_l}^2}{x_l-x_i} \\
\nonumber & + \frac{3(h+1)}{32} \sum_{l\neq i,j,j+1} \frac{[(5+h)q_i-(5h+1)q_l] \partial_{x_l} - 4[2h q_i+(h+1)q_l] \partial_{x_i}}{(x_l-x_i)^2} \\
\nonumber & - \frac{3(h+1)^2}{8} \sum_{l \neq i,j,j+1} \sum_{k \neq l,i,j,j+1} \frac{q_l\partial_{x_k}}{(x_i-x_l)(x_k-x_l)} - \frac{h(h+1)(h+5)}{16} \sum_{l \neq i,j,j+1} \frac{q_l+3q_i}{(x_l-x_i)^3} \\
\nonumber & + \frac{3h(h+1)^2}{8} \sum_{l \neq i,j,j+1} \sum_{k \neq l,i,j,j+1}  \frac{q_l}{(x_i-x_l)(x_k-x_l)^2} = 0,
\end{align}
where $\mathcal C$ is a certain differential operator. Note that all terms not proportional to $f_0$ and $f_2$ correspond to the differential operator $\mathcal D_{\varsigma,c}^{(i)} f = 0$ in \eqref{defoperatorBPZ} where all terms containing $x_j$ and $x_{j+1}$ have been removed. The remaining part of the proof then consists of showing that the terms $\mathcal C f_0 - \frac{9(h+1)(h-1)f_2}{2(\hat x_j-x_m)}$ contribute to such missing terms. This is done by excluding the terms corresponding to $i=m$ in the sums in \eqref{relationf2f0}, substituting the result into \eqref{eqc2m}, and simplifying all the terms. 


\end{proof}

\begin{remark}
    The indicial equation \eqref{indicialeq} is automatically satisfied because a suitable exponent has been chosen in the series expansion \eqref{frobeniusansatz}. More generally, if the exponent $-h$ is replaced by some unknown $X$, then the indicial equation reads
    \begin{equation}
        -X(X-1)(X-2) + \frac{3X(h+1)^2}{4} - \frac{h(h+1)(h+5)}{4} = 0.
    \end{equation}
    This equation has solutions $X_1 = -h$, $X_2 = \frac{1+h}2$ and $X_3 = \frac{5+h}2$. 
\end{remark}

The next Lemma gives a definition for the conformal block functions $\mathcal U_T(x_1,\dots,x_d)$ that is equivalent to Definition \ref{defconformalblocks}. This will enable us to apply Lemma \ref{lemmafusion} recursively to prove \eqref{W3BPZ} for any $\varsigma$.

\begin{lemma} \label{lemmasecondrepCBs}
    The conformal block functions $\mathcal U_T(x_1,\dots,x_d)$ of Definition \ref{defconformalblocks} possess the following equivalent definition:
\begin{equation} \label{secondrepresentationconformalblocks}
    \mathcal U_T(x_1,\dots,x_d) = \left[\lb\prod_{k=1}^d \prod_{p_k \leq i < j < p_{k+1}} (x_j-x_i)^\frac13\rb \mathcal U_{\tilde T}(x_1,\dots,x_n)\right]_\text{eval},
\end{equation}
where eval is the assignment map defined in Definition \ref{remarkeval}, and where $\tilde T$ is defined from $T$ as in Definition \ref{defTtilde}. 
\end{lemma}
\begin{proof}
    Since $\tilde T$ is a numbering (every entry in the tableau $\tilde T$ appear exactly once), we can write
    \begin{equation}
        \mathcal U_T(x_1,\dots,x_d) = \left[R(x_1,\dots,x_n) \mathcal P_{\widetilde{T^t}}(x_1,\dots,x_n)\right]_\text{eval},
    \end{equation}
    with 
    \begin{equation}
      R(x_1,\dots,x_n) = \frac{\lb\prod_{k=1}^d \prod_{p_k \leq i < j < p_{k+1}} (x_j-x_i)^\frac13\rb}{\prod_{1\leq i < j \leq n} (x_j-x_i)^\frac13}.
    \end{equation}
    First, we have $\left[ \mathcal P_{\widetilde{T^t}}(x_1,\dots,x_n) \right]_\text{eval} = \mathcal P_T(x_1,\dots,x_d)$. It then remains to show that $\left[R(x_1,\dots,x_n)\right]_\text{eval}$ is finite and is given by the prefactor in \eqref{defconformalblocks}. This is readily observed by considering the three cases (i) $s_i=1$ and $s_j=1$ for some $i \in \llbracket 1,\dots,d \rrbracket$ with $i\neq j$, (ii) $s_i=1$ and $s_j=2$ or $s_i=2$ and $s_j=1$, and (iii) $s_i=2$ and $s_j=2$.
\end{proof}

\begin{lemma} \label{lemmaanalytic}
    Let $f: \mathcal H_{n+1} \to \mathbb R$ such that $f(x_1,\cdots,x_{n+1}) = P_1(x_1,\cdots,x_{n+1}) P_2(x_1,\cdots,x_{n+1})^{-\frac13}$, where $P_1$ and $P_2$ are polynomials. Then, for any $j = 1,...,n$ there exists $\eta_{j} \in \mathbb Z$ such that
    $$f(x_1,\cdots,x_{n+1}) = (x_{j+1}-x_j)^{\frac{\eta_j}2} \sum_{k=0}^\infty u_k (x_{j+1}-x_j)^k.$$
    Moreover, all coefficients $u_k$ are of the form $P_k Q_k^{-\frac13}$ where $P_k$, $Q_k$ are polynomials in the variables $x_1,\cdots,x_j,x_{j+2},\cdots,x_{n+1}$.
\end{lemma}

\begin{proof}
The first claim readily follows since for any $j = 1,..., n$ there exists $\eta_{j} \in \mathbb Z$ such that $f(x_1,\cdots,x_{n+1}) = (x_{j+1}-x_j)^{\frac{\eta_j}2} g(x_1,\cdots,x_{n+1})$, where $g$ is (real) analytic at $x_{j+1}=x_j$. As for the second claim, we have
$$u_k = \frac{1}{k!} \left[\partial_{x_{j+1}}^k \lb(x_{j+1}-x_j)^{-\frac{\eta_j}2}  f(x_1,\cdots,x_{n+1})\rb \right]_{x_{j+1}=x_j}.$$
Thus any $u_k$ is of the form $P_k Q_k^{-\frac13}$ where $P$ and $Q$ are polynomials in the variables $x_1,\cdots,x_j,x_{j+2},\cdots,x_{n+1}$. This completes the proof. 
\end{proof}

We are finally ready to prove Proposition \ref{propeqW3forUT}.

\begin{proof}[Proof of Proposition \ref{propeqW3forUT}]
    Let $T \in \RSYT$ and $\tilde T \in \NB$ be its associated numbering according to Definition \ref{defTtilde}. We start from the case $\varsigma = (1^n)$ and proceed by induction. By Lemma \ref{lemmaW3BPZvarsigmaequals1}, we have that $\mathcal U_{\tilde T}(x_1,\dots,x_n)$ satisfies
    $$\mathcal D^{(j)}_{(1^n),2} \; \mathcal U_{\tilde T}(x_1,\cdots,x_n) = 0, \qquad j=1,\cdots,n.$$
    Now, let $m \in \llbracket 2,d \rrbracket$ as well as $\varsigma_{m-1}=(s_1,\dots,s_{m-1},1,\dots,1)$ such that $|\varsigma_{m-1}|=n$, and suppose that the function $g_{m-1}$ defined by
    \begin{equation}
        g_{m-1} = \lim\limits_{\substack{x_{p_{m-1}}\to x_{m-1}\\x_{p_{m-1}+s_{m-1}-1}\to x_{m-1}}} \cdots \lim\limits_{x_{1+s_{1}-1} \to x_1} \left[\lb\prod_{k=1}^{m-1} \prod_{p_k \leq i < j < p_{k+1}} (x_j-x_i)^\frac13\rb \mathcal U_{\tilde T}(x_1,\cdots,x_n)\right]
    \end{equation}
    satisfies
    $$\mathcal D^{(j)}_{\varsigma_{m-1},2} \; g_{m-1} = 0, \qquad j=1,\dots,m-1,p_m,\dots,n.$$
   It then follows from Lemma \ref{lemmafusion} that the function $g_m$ satisfies
    $$\mathcal D^{(j)}_{\varsigma_m,2} \; g_m = 0, \qquad j=1,\dots,m,p_{m+1},\dots,n.$$
    Note that we can apply Lemma \ref{lemmafusion} here thanks to Lemma \ref{lemmaanalytic}, since $\mathcal U_{\tilde T}(x_1,\cdots,x_n)$ is of the form $P_1 P_2^{-\frac13}$ where $P_1$ and $P_2$ are polynomials. By induction, we conclude that the function 
    $$g_d = \mathcal U_T(x_1,\dots,x_d)$$
    satisfies
    $$\mathcal D^{(j)}_{\varsigma,2} \; g_d = 0, \qquad j=1,\dots,d,$$
    which is \eqref{eqW3BPZforUT}. This completes the proof.
\end{proof}

\subsection{Proof of \eqref{covariance}}

We now turn to the covariance properties satisfied by $\mathcal U_T$.

\begin{proposition} \label{propcovariance}
Let $\varphi(x) = \frac{ax+b}{cx+d}$ with $ad-bc \neq 0$ and such that $\varphi(x_1) < \varphi(x_2) < \cdots < \varphi(x_n)$. Moreover, let $T \in \textup{RSYT}^{\pi}_\varsigma$. Then we have
\begin{equation} \label{covarianceCB} 
     \mathcal U_T(\varphi(x_1),\cdots ,\varphi(x_d)) = \lb \prod_{i=1}^d \varphi'(x_i)^{-\frac13} \rb \mathcal U_T(x_1,\cdots ,x_d).
\end{equation}
\end{proposition}
\begin{proof}
    We rewrite the representation \eqref{definitionCB} in a similar way as in \eqref{dePTfphi}:
    \begin{equation} \label{representationUTalpha}
         \mathcal U_T(x_1,\cdots,x_d) = \lb \prod_{1\leq i<j\leq d} (x_j-x_i)\rb ^{\alpha_T(i,j)},
    \end{equation}
where 
\begin{equation} \label{defalphaT}
   \alpha_T(i,j) = \psi_{T^t}(i,j)-\frac{s_i s_j}{3},
\end{equation}
and where $\psi_{T^t}(i,j)$ is the number of (unordered) pairs of entries $\{i,j\}$, $i\neq j$, of $T^t$ appearing in the same column. Thus we have $\psi_{T^t}(j,i)=\psi_{T^t}(i,j)$ and $\psi_{T^t}(i,i)=0$. Utilizing the crucial identity \cite[Lemma 4.7]{KP2}
\begin{equation}
    \varphi(x)-\varphi(y) = (x-y) \sqrt{\varphi'(x) \varphi'(y)},
\end{equation}
we compute
\begin{align*}
    \mathcal U_T(\varphi(x_1),\cdots,\varphi(x_d)) & = \lb \prod_{1\leq i<j\leq d} (\varphi'(x_j) \varphi'(x_i)) ^{\frac{\psi_{T^t}(i,j)}2-\frac{s_i s_j}{6}}\rb \mathcal U_T(x_1,\cdots,x_d)  \\
    & =  \lb \prod_{j=1}^d \prod_{i=1}^{j+1} \varphi'(x_j)^{\frac{\psi_{T^t}(i,j)}2-\frac{s_i s_j}{6}}\rb \lb \prod_{i=1}^d \prod_{j=i+1}^d \varphi'(x_i)^{\frac{\psi_{T^t}(i,j)}2-\frac{s_i s_j}{6}}\rb \mathcal U_T(x_1,\cdots,x_d) \\
    & = \lb \prod_{j=1}^d \varphi'(x_j)^{\sum_{i=1}^{j+1}\lb\frac{\psi_{T^t}(i,j)}2-\frac{s_i s_j}{6} \rb}\rb \lb \prod_{i=1}^d \varphi'(x_i)^{\sum_{j=i+1}^d\lb\frac{\psi_{T^t}(i,j)}2-\frac{s_i s_j}{6}\rb}\rb \mathcal U_T(x_1,\cdots,x_d).
\end{align*}
Recalling that $\sum_{i=1}^d s_i = n$, we then obtain
\begin{align*}
    \mathcal U_T(\varphi(x_1),\cdots,\varphi(x_d)) & = \lb \prod_{j=1}^d \varphi'(x_j)^{\frac{\sum_{i=1}^d \psi_{T^t}(i,j)}2-\frac{s_j(n-s_j)}{6}}\rb \mathcal U_T(x_1,\cdots,x_d).
    \end{align*}
We now have the key identity
\begin{equation} \label{identitysumpsi}
    \sum_{i=1}^d \psi_{T^t}(i,j) = s_j\lb \frac{n}3-1 \rb,
\end{equation}
since $T^t \in \text{CSYT}^{\Bar{\pi}}_\varsigma$, and there are $\frac{n}3-1$ linear factors containing $x_j$ in each column of $T^t$ containing the entry $j$. We conclude that
\begin{align*}
    \mathcal U_T(\varphi(x_1),\cdots,\varphi(x_d)) & = \lb \prod_{j=1}^d \varphi'(x_j)^{\frac{s_j(s_j-3)}6}\rb \mathcal U_T(x_1,\cdots,x_d),
    \end{align*}
and for each $s_j=1$ and $s_j=2$, we have $\frac{s_j(s_j-3)}6 = -\frac13$. This completes the proof.
\end{proof}

\begin{remark}
    The identity \eqref{identitysumpsi}, which holds only if $T$ has rectangular shape, is the reason why Proposition \ref{propcovariance} holds only for $T$ having a rectangular shape. If $T$ does not have a rectangular shape, \eqref{globalward1} still holds, however, we expect that \eqref{globalward2} is modified and that \eqref{globalward3} does not hold at all.
\end{remark}

\subsection{Proof of \eqref{ward1}-\eqref{ward5}}

\begin{proposition} \label{propW3WardUT}
    Let $T \in \textup{CSYT}^{\pi}_\varsigma$. We have
        \begin{align} \label{eqMmUT}
    \mathcal M^{(m)}_{\varsigma,c} \; \mathcal U_T(x_1,\cdots ,x_d) = 0, \qquad m=1,2,3,4,5,
\end{align}
where $\mathcal M^{(m)}_{\varsigma,c}$ is the differential operator defined in \eqref{defM}.
\end{proposition}
For the proof of Proposition \ref{propW3WardUT} we need the following two Lemmas:
\begin{lemma} \label{lemmasmartidentities}
Let $\alpha_T(i,j)$ be defined in \eqref{defalphaT}. Moreover, define the functions (see the proof of Proposition \ref{propW3WardUT} for the choice of notations)
  \begin{align}
   \label{defAm} & A^{(m)}(i,j,k) = q_j x_j^{m-1} \alpha_T(i,j)\alpha_T(j,k) + \frac13 q_i x_i^{m-1} \alpha_T(j,k) + \frac13 q_k x_k^{m-1} \alpha_T(i,j), \quad m=1,\dots,5, \\
   \label{defBm}  & B^{(m)}(i,j) = q_j x_j^{m-1} \lb \alpha_T(i,j)^2 - \frac13 \alpha_T(i,j) - \frac29 \rb + \frac{m-1}3 q_j x_j^{m-2} \alpha_T(i,j) (x_j-x_i), \quad m=1,\dots,5, \\
  \label{defY} & Y(i,j,k) = \frac{q_j}3\lb\alpha_T(i,j) \alpha_T(j,k) - \frac13 \alpha_T(i,k)\rb, \\
    & \hat B^{(3)}(i,j) = B^{(3)}(i,j) - \frac13 \alpha_T(i,j) q_j (x_j-x_i)^2, \\
    & \hat B^{(4)}(i,j) = B^{(4)}(i,j) - \frac13 \alpha_T(i,j)x_j (q_i+2q_j) (x_j-x_i)^2, \\
    & \hat B^{(5)}(i,j) = B^{(5)}(i,j) - \frac13 \alpha_T(i,j) xj(x_3+\frac{x_i}2)(q_i+3q_j) (x_j-x_i)^2.
  \end{align}
  The following list of identities holds:
  \begin{align}
     \label{identityalphasquares} & \alpha_T(i,j)^2 = \frac29 + \frac{q_i q_j}3 \alpha_T(i,j).\\
    \label{B12antisymij}  & B^{(m)}(i,j) + B^{(m)}(j,i) = 0, \qquad m=1,2, \\
    \label{hatB345antisymij} & \hat B^{(m)}(i,j) + \hat B^{(m)}(j,i) = 0, \qquad m=3,4,5, \\
       \label{A1simmetricij} & A^{(1)}(i,j,k) - A^{(1)}(j,i,k) = 0, \\
    \label{Ysymij} & Y(i,j,k) - Y(j,i,k) = 0, 
  \end{align}
\end{lemma}
\begin{proof}
To prove \eqref{identityalphasquares}, we proceed case by case. Up to permutations, there are three cases to consider. 

\textbf{First case}: $s_i=s_j=1$ implying $q_i=q_j=1$. Either $i$ and $j$ are the same column or not. In both cases \eqref{identityalphasquares} follows.

\textbf{Second case}: $s_i=2,\;s_j=1$ implying $q_i=-1,\;=q_j=1$. Either $i$ and $j$ are the same column or not. In both cases \eqref{identityalphasquares} follows.

\textbf{Third case}: $s_i=s_j=2$ implying $q_i=q_j=-1$. Either $i$ and $j$ are the same two columns or not. In both cases \eqref{identityalphasquares} follows.

Moreover, direct computations utilizing \eqref{identityalphasquares} show that \eqref{B12antisymij} for $m=1,2$ and \eqref{hatB345antisymij} for $m=3,4,5$ hold.

Finally, \eqref{A1simmetricij} and \eqref{Ysymij} are proved in the same way as \eqref{identityalphasquares}. Up to permutations, there are four cases to consider. 
    
    \textbf{First case}: $s_i=s_j=s_k=1$ implying $q_i=q_j=q_k=1$. Either $i,j,k$ are in different columns or exactly two of them are in the same column or all of them are in the same column. In every case the identities follow.
    
    \textbf{Second case}: $s_i=2,\;s_j=s_k=1$ implying $q_i=-1,\;q_j=q_k=1$. Either $j,k$ are in the same column which can be containing $i$ or not. Or $j,k$ are not in the same column and they there can be exactly one or two of them in a column containing $i$. In every case the identities follow.
    
    \textbf{Third case}: $s_i=s_j=2,\;s_k=1$ implying $q_i=q_j=-1,\;q_k=1$. Either $i,j$ are in the same two colums and $k$ can be in a column containing $i$ and $j$ or not. Or $i,j$ share only one column and $k$ can be in this column or not. In every case the identities follow.

    \textbf{Fourth case}: $s_i=s_j=s_k=2$ implying $q_i=q_j=q_k=-1$. Either $i,j$ are in the same two columns and $k$ appears once or twice in these two columns. Or $i,j$ share only one column and $k$ appears in this column or not. In every case the identities follow.
\end{proof}
\begin{lemma} \label{lemmasumidentities}
Let $T \in \textup{RSYT}^{\pi}_\varsigma$. The following sum identities hold:
\begin{align}
  \label{sumidentitymeq3} & \sum_{j=1}^d \sum_{i \neq j}^d \sum_{k \neq i,j}^d Y(i,j,k) + \frac13 \sum_{j=1}^d q_j \sum_{i \neq j}^d \alpha_T(i,j) + \frac{2}{27}\sum_{j=1}^d q_j = 0, \\
   \label{sumidentitymeq4}  & 3\sum_{j=1}^d \sum_{i \neq j}^d \sum_{k \neq i,j}^d x_j Y(i,j,k) + \frac13 \sum_{j=1}^d \sum_{i \neq j}^d \alpha_T(i,j)x_j (q_i+2q_j) + \frac29 \sum_{j=1}^d q_j x_j = 0, \\
    \label{sumidentitymeq5} & 3\sum_{j=1}^d \sum_{i \neq j}^d \sum_{k \neq i,j}^d (x_j^2+\frac{x_j}2(x_i+x_k)) Y(i,j,k) + \frac13 \sum_{j=1}^d \sum_{i \neq j}^d \alpha_T(i,j) x_j(x_j+\frac{x_i}2)(q_i+3q_j) + \frac49 \sum_{j=1}^d q_j x_j^2 = 0.
\end{align}
\end{lemma}
\begin{proof}
We record some important facts for the proof of Lemma \ref{lemmasumidentities}. Recall from \eqref{defalphaT} that $\alpha_T(i,j) = \psi_{T^t}(i,j) - \frac{s_i s_j}3$, with $\psi_{T^t}(i,j) = \psi_{T^t}(j,i)$ and $\psi_{T^t}(i,i) = 0$. Moreover, utilizing \eqref{identitysumpsi} it is easily showed that
\begin{align}
   & \sum_{k=1}^d \alpha_T(j,k) = - s_j, \\
 \label{sumkdiffiandjofalpha}  & \sum_{k \neq i,j}^d \alpha_T(j,k) = -s_j - \alpha_T(i,j) + \frac{s_j^2}3.
\end{align}

\subsubsection*{Proof of \eqref{sumidentitymeq3}} Straightforward computations utilizing \eqref{sumkdiffiandjofalpha} show that
\begin{equation} \label{firstsumofY}
    \sum_{k \neq i,j} Y(i,j,k) = \frac{q_j}3 \lb - \alpha_T(i,j)^2 -s_j \alpha_T(i,j) + \alpha_T(i,j) \frac{s_j^2}3 + \frac{s_i}3 + \frac{\alpha_T(i,j)}3 - \frac{s_i^2}9 \rb.
\end{equation}
Substitution of the identities \eqref{identityalphasquares} and
\begin{equation} \label{sjsquared}
    s_j^2 = 3s_j-2
\end{equation}
into \eqref{firstsumofY} leads to 
\begin{equation} \label{sumofY}
    \sum_{k \neq i,j} Y(i,j,k) = \frac{2q_j}9 \lb -5 + s_i(3-2s_j) + 3s_j\rb \alpha_T(i,j).
\end{equation}
Next, thanks to 
\begin{equation}
    s_j = \frac{3-q_j}2
\end{equation}
and $q_j^2 = 1$, it can be showed that
\begin{equation}
   \sum_{k \neq i,j} Y(i,j,k) + \frac{q_j \alpha_T(i,j)}3 = - \frac{\alpha_T(i,j)}3 + \frac29 q_j \alpha_T(i,j) + \frac29 s_i \alpha_T(i,j).
\end{equation}
A straightforward computation then leads to
\begin{equation}
  \sum_j \sum_{i \neq j} \left[\sum_{k \neq i,j} Y(i,j,k) + \frac{q_j \alpha_T(i,j)}3\right] = \sum_j \lb \frac{s_j}3 - \frac29 q_j s_j -\frac29 s_j^2 - \frac{s_j^2}9 + \frac{2}{27} q_j s_j^2 + \frac{2}{27} s_j^3 \rb.
\end{equation}
It finally remains to use $q_j = 3-2s_j$ and $s_j^2 = 3s_j-2$ to obtain
\begin{equation}
  \sum_j \sum_{i \neq j} \left[\sum_{k \neq i,j} Y(i,j,k) + \frac{q_j \alpha_T(i,j)}3\right] = \sum_j \frac{2}{27}(2s_j-3) = -\frac{2}{27} \sum_j q_j,
\end{equation}
which is \eqref{sumidentitymeq3}.

\subsubsection*{Proof of \eqref{sumidentitymeq4}} The proof of \eqref{sumidentitymeq4} involves similar computations, hence we will skip more steps. In particular, we can immediately use \eqref{sumofY} to obtain
\begin{align*}
    & \sum_j \sum_{i \neq j} \left[\sum_{k \neq i,j} (3x_j Y(i,j,k)) + \frac13 \alpha_T(i,j) x_j (q_i+2q_j) \right] \\
    & = \sum_j \sum_{i \neq j} \left[x_j \frac{2q_j}3 (-5+s_i(3-2s_j)+3s_j)\alpha_T(i,j)+ \frac13 \alpha_T(i,j) x_j (q_i+2q_j) \right].
\end{align*}
Then, computations that are similar to before show that this equals 
\begin{equation}
    \sum_j \sum_{i \neq j} \frac{q_j x_j}3 \alpha_T(i,j).
\end{equation}
After further simplifications we then obtain
\begin{equation}
    \sum_j \sum_{i \neq j} \left[\sum_{k \neq i,j} (3x_j Y(i,j,k)) + \frac13 \alpha_T(i,j) x_j (q_i+2q_j) \right] = -\frac29 \sum_j q_j x_j,
\end{equation}
which is \eqref{sumidentitymeq4}.

\subsubsection*{Proof of \eqref{sumidentitymeq5}}
The identity $Y(i,j,k) = Y(k,j,i)$ implies that the triple sum in the left-hand side of \eqref{sumidentitymeq5} equals
\begin{equation}
    3\sum_{j=1}^d \sum_{i \neq j}^d \sum_{k \neq i,j}^d (x_j^2+\frac{x_j}2(x_i+x_k)) Y(i,j,k) = 3\sum_{j=1}^d \sum_{i \neq j}^d \sum_{k \neq i,j}^d (x_j^2+ x_j x_i) Y(i,j,k).
\end{equation}
We then perform computations that are similar as before to obtain
\begin{align}
  &  \sum_{j=1}^d \sum_{i \neq j}^d \left[\sum_{k \neq i,j} (3 (x_j^2+ x_j x_i) Y(i,j,k)) + \frac13 \alpha_T(i,j) xj(x_j+\frac{x_i}2)(q_i+3q_j) \right] \\
  & = \sum_{j=1}^d \sum_{i \neq j}^d \frac13 x_i x_j \alpha_T(i,j) (s_i-s_j) + \sum_{j=1}^d \sum_{i \neq j}^d \lb 2x_j^2 \alpha_T(i,j) - \frac43 s_j x_j^2 \alpha_T(i,j) \rb.
\end{align}
The first double sum vanishes by antisymmetry of the summand, and the second double sum reduces to
\begin{equation*}
    \sum_{j=1}^d \sum_{i \neq j}^d \lb 2x_j^2 \alpha_T(i,j) - \frac43 s_j x_j^2 \alpha_T(i,j) \rb = \sum_{j=1}^d x_j^2 \lb -2s_j + 2 s_j^2 - \frac49 s_j^3 \rb = \sum_j \frac{4}9 x_j^2 (2s_j-3) = -\frac49 \sum_j q_j x_j^2.
\end{equation*}
This concludes the proof of \eqref{sumidentitymeq5} and of Lemma \ref{lemmasumidentities}.
\end{proof}
With Lemmas \ref{lemmasmartidentities} and \ref{lemmasumidentities} at hand, we are now ready to prove Proposition \ref{propW3WardUT}.

\begin{proof}[Proof of Proposition \ref{propW3WardUT}]
For each $m=1,2,3,4,5$, straightforward computations show that the action of $\mathcal M^{(m)}_{\varsigma,c}$ on the formula \eqref{representationUTalpha} for the conformal block function $\mathcal U_T$ takes the following form:
\begin{align} \label{actionofMonU}
   \frac{\mathcal M^{(m)}_{\varsigma,c} \; \mathcal U_T(x_1,\dots,x_d)}{\mathcal U_T(x_1,\dots,x_d)} = \sum_{j=1}^d \sum_{i \neq j}^d \sum_{k \neq i,j}^d \frac{A^{(m)}(i,j,k)}{(x_j-x_i)(x_j-x_k)} + \sum_{j=1}^d \sum_{i \neq j}^d \frac{B^{(m)}(i,j)}{(x_j-x_i)^2} + \sum_{j=1}^d C^{(m)}(j),
\end{align}
where $A^{(m)}$ and $B^{(m)}$ are defined in \eqref{defAm} and \eqref{defBm}, respectively, and
\begin{align}
C^{(m)}(j) = \frac{(m-1)(m-2)}{27} q_j x_j^{m-3}.
\end{align}
We now show that \eqref{actionofMonU} equals 0 for each $m=1,2,3,4,5$.
\subsubsection*{The case m=1.} By \eqref{B12antisymij} for $m=1$, the double sum vanishes. We now totally symmetrize the summand on the triple sum with respect to all permutations of $i,j,k$. Utilizing $A^{(1)}(i,j,k)=A^{(1)}(k,j,i)$, we have
\begin{align*}
    & \sum_{j=1}^d \sum_{i \neq j}^d \sum_{k \neq i,j}^d \frac{A^{(1)}(i,j,k)(x_i-x_k)}{(x_j-x_i)(x_j-x_k)(x_i-x_k)} \\
    & = \frac13 \sum_{j=1}^d \sum_{i \neq j}^d \sum_{k \neq i,j}^d \frac{A^{(1)}(i,j,k)(x_i-x_k) + A^{(1)}(k,i,j)(x_k-x_j) + A^{(1)}(j,k,i)(x_j-x_i)}{(x_j-x_i)(x_j-x_k)(x_i-x_k)}.
\end{align*}
By \eqref{A1simmetricij}, we also have $A^{(1)}(i,j,k)=A^{(1)}(j,k,i)=A^{(1)}(k,i,j)$. Hence the triple sum vanishes, since the numerator in the summand vanishes. This concludes the proof of \eqref{eqMmUT} for $m=1$.

\subsubsection*{The case m=2.} By \eqref{B12antisymij} for $m=2$, the double sum vanishes. As for the triple sum, we proceed in the same way as for the case $m=1$. A straightforward computation shows that
\begin{align*}
   & A^{(2)}(i,j,k)(x_i-x_k) + A^{(2)}(k,i,j)(x_k-x_j) + A^{(2)}(j,k,i)(x_j-x_i) \\
   & = 3 \left[x_j(x_i-x_k) Y(i,j,k) + x_k(x_j-x_i) Y(j,k,i) + x_i(x_k-x_j) Y(k,i,j)\right],
\end{align*}
where $Y$ is defined in \eqref{defY}. Moreover, utilizing the symmetries $Y(i,j,k)=Y(k,j,i)$ and \eqref{Ysymij} we obtain
\begin{align*}
  A^{(2)}(i,j,k)(x_i-x_k) & + A^{(2)}(k,i,j)(x_k-x_j) + A^{(2)}(j,k,i)(x_j-x_i) \\
  & = 3 Y(i,j,k)\left[x_i(x_k-x_j)+x_k(x_j-x_i)+x_j(x_i-x_k)\right] \\
  & = 0.
\end{align*}
This shows that the symmetrization of the summand in the triple sum vanishes. This concludes the proof of \eqref{eqMmUT} for $m=2$.
\subsubsection*{The case m=3.} The case $m=3$ (just like the cases $m=4$ and $m=5$) is more intricate than the cases $m=1$ and $m=2$. We first study the double sum. The identity \eqref{hatB345antisymij} for $m=3$ immediately leads to
\begin{equation}
   \sum_{j=1}^d \sum_{i \neq j}^d \frac{B^{(3)}(i,j)}{(x_j-x_i)^2} = \frac13 \sum_{j=1}^d q_j \sum_{i \neq j}^d \alpha_T(i,j).
\end{equation}
We now investigate the triple sum. The key is to introduce a function $\hat A^{(3)}$ such that
$$\frac{A^{(3)}(i,j,k)}{(x_j-x_i)(x_j-x_k)} = \frac{\hat A^{(3)}(i,j,k)}{(x_j-x_i)(x_j-x_k)(x_i-x_k)} + Y(i,j,k).$$
Since $A^{(3)}(i,j,k) = A^{(3)}(k,j,i)$ and $Y(i,j,k) = Y(k,j,i)$, it follows that $\hat A^{(3)}(i,j,k) = - \hat A^{(3)}(k,j,i)$. We then have
\begin{equation}
    \sum_{j=1}^d \sum_{i \neq j}^d \sum_{k \neq i,j}^d \frac{\hat A^{(3)}(i,j,k)}{(x_j-x_i)(x_j-x_k)(x_i-x_k)} = \frac13 \sum_{j=1}^d \sum_{i \neq j}^d \sum_{k \neq i,j}^d \frac{\hat A^{(3)}(i,j,k) + \hat A^{(3)}(k,i,j) + \hat A^{(3)}(j,k,i)}{(x_j-x_i)(x_j-x_k)(x_i-x_k)}.
\end{equation}
A straightforward computation leads to
\begin{align*}
  &  \hat A^{(3)}(i,j,k) + \hat A^{(3)}(k,i,j) + \hat A^{(3)}(j,k,i) = Y(i,j,k) (X^{(3)}(i,j,k) + X^{(3)}(k,i,j) + X^{(3)}(j,k,i)) = 0,
\end{align*}
where
\begin{align*}
  &  X^{(3)}(i,j,k) = (x_k-x_i)(x_j(x_i+x_k)-x_ix_k+2x_j^2).
\end{align*}
We then infer that
\begin{equation}
\frac{\mathcal M^{(3)}_{\varsigma,c} \; \mathcal U_T(x_1,\dots,x_d)}{\mathcal U_T(x_1,\dots,x_d)} = \sum_{j=1}^d \sum_{i \neq j}^d \sum_{k \neq i,j}^d Y(i,j,k) + \frac13 \sum_{j=1}^d q_j \sum_{i \neq j}^d \alpha_T(i,j) + \frac{2}{27}\sum_{j=1}^d q_j,
\end{equation}
and it equals 0 by Equation \eqref{sumidentitymeq3} in Lemma \ref{lemmasumidentities}. This concludes the case $m=3$.
\subsubsection*{The case m=4.} We proceed similarly to the case $m=3$. The identity \eqref{hatB345antisymij} for $m=4$ leads to
\begin{equation}
   \sum_{j=1}^d \sum_{i \neq j}^d \frac{B^{(4)}(i,j)}{(x_j-x_i)^2} = \frac13 \sum_{j=1}^d \sum_{i \neq j}^d \alpha_T(i,j)x_j (q_i+2q_j).
\end{equation}
As for the triple sum, we define $\hat A^{(4)}$ such that
$$\frac{A^{(4)}(i,j,k)}{(x_j-x_i)(x_j-x_k)} = \frac{\hat A^{(4)}(i,j,k)}{(x_j-x_i)(x_j-x_k)(x_i-x_k)} + 3 x_j Y(i,j,k).$$
We then have
\begin{equation}
    \sum_{j=1}^d \sum_{i \neq j}^d \sum_{k \neq i,j}^d \frac{\hat A^{(4)}(i,j,k)}{(x_j-x_i)(x_j-x_k)(x_i-x_k)} = \frac13 \sum_{j=1}^d \sum_{i \neq j}^d \sum_{k \neq i,j}^d \frac{\hat A^{(4)}(i,j,k) + \hat A^{(4)}(k,i,j) + \hat A^{(4)}(j,k,i)}{(x_j-x_i)(x_j-x_k)(x_i-x_k)}.
\end{equation}
A straightforward computation leads to
\begin{align*}
  &  \hat A^{(4)}(i,j,k) + \hat A^{(4)}(k,i,j) + \hat A^{(4)}(j,k,i) = Y(i,j,k) (X^{(4)}(i,j,k) + X^{(4)}(k,i,j) + X^{(4)}(j,k,i)) = 0,
\end{align*}
where
\begin{align*}
  &  X^{(4)}(i,j,k) = 3 x_j (x_i-x_k) (x_i x_j-x_i x_k+x_j x_k).
\end{align*}
We then infer that
\begin{equation}
\frac{\mathcal M^{(4)}_{\varsigma,c} \; \mathcal U_T(x_1,\dots,x_d)}{\mathcal U_T(x_1,\dots,x_d)} = 3\sum_{j=1}^d \sum_{i \neq j}^d \sum_{k \neq i,j}^d x_j Y(i,j,k) + \frac13 \sum_{j=1}^d \sum_{i \neq j}^d \alpha_T(i,j)x_j (q_i+2q_j) + \frac29 \sum_{j=1}^d q_j x_j,
\end{equation}
which vanishes by Equation \eqref{sumidentitymeq4} in Lemma \ref{lemmasumidentities}. This completes the case $m=4$.
\subsubsection*{The case m=5.} Again, we proceed similarly to the cases $m=3$ and $m=4$. The identity \eqref{hatB345antisymij} for $m=5$ leads to
\begin{equation}
   \sum_{j=1}^d \sum_{i \neq j}^d \frac{B^{(5)}(i,j)}{(x_j-x_i)^2} = \frac13 \sum_{j=1}^d \sum_{i \neq j}^d \alpha_T(i,j) xj(x_3+\frac{x_i}2)(q_i+3q_j) (x_j-x_i)^2.
\end{equation}
As for the triple sum, we define $\hat A^{(5)}$ such that
$$\frac{A^{(5)}(i,j,k)}{(x_j-x_i)(x_j-x_k)} = \frac{\hat A^{(5)}(i,j,k)}{(x_j-x_i)(x_j-x_k)(x_i-x_k)} + 3 \lb x_j^2+\frac{x_j}2(x_i+x_k)\rb Y(i,j,k).$$
We then have
\begin{equation}
    \sum_{j=1}^d \sum_{i \neq j}^d \sum_{k \neq i,j}^d \frac{\hat A^{(5)}(i,j,k)}{(x_j-x_i)(x_j-x_k)(x_i-x_k)} = \frac13 \sum_{j=1}^d \sum_{i \neq j}^d \sum_{k \neq i,j}^d \frac{\hat A^{(5)}(i,j,k) + \hat A^{(5)}(k,i,j) + \hat A^{(5)}(j,k,i)}{(x_j-x_i)(x_j-x_k)(x_i-x_k)}.
\end{equation}
A straightforward computation leads to
\begin{align*}
  &  \hat A^{(5)}(i,j,k) + \hat A^{(5)}(k,i,j) + \hat A^{(5)}(j,k,i) = Y(i,j,k) (X^{(5)}(i,j,k) + X^{(5)}(k,i,j) + X^{(5)}(j,k,i)) = 0,
\end{align*}
where
\begin{align*}
  &  X^{(5)}(i,j,k) = \frac{3}{2} x_j \left(x_i-x_k\right) \left(x_i^2 \left(x_j-x_k\right)+x_j x_k \left(x_j+x_k\right)+x_i \left(x_j^2-x_k^2\right)\right).
\end{align*}
We then infer that
\begin{align}
\frac{\mathcal M^{(5)}_{\varsigma,c} \; \mathcal U_T(x_1,\dots,x_d)}{\mathcal U_T(x_1,\dots,x_d)} = \; & 3\sum_{j=1}^d \sum_{i \neq j}^d \sum_{k \neq i,j}^d \lb x_j^2+\frac{x_j}2(x_i+x_k)\rb Y(i,j,k) \\
\nonumber & + \frac13 \sum_{j=1}^d \sum_{i \neq j}^d \alpha_T(i,j) x_j(x_j+\frac{x_i}2)(q_i+3q_j) + \frac49 \sum_{j=1}^d q_j x_j^2.
\end{align}
This vanishes by Equation \eqref{sumidentitymeq5} in Lemma \ref{lemmasumidentities}. This completes the case $m=5$ and thus the proof of Proposition \ref{propW3WardUT}.
\end{proof}

\begin{remark}
    We note that the first two Ward identities \eqref{eqMmUT} for $m=1,2$ hold even if $T$ does not have a rectangular shape. However, the last three identities \eqref{eqMmUT} for $m=3,4,5$ hold if and only if $T$ has a rectangular shape, because of Lemma \ref{lemmasumidentities}.
\end{remark}

\subsection{Asymptotic properties of the conformal block basis}
\label{sectionasy}

In this section we study asymptotic properties of the $W_3$ conformal block basis elements $\{\mathcal{U}_T,\; T\in \text{RSYT}_\varsigma^{\pi}\}$.

\begin{proposition}
    Let $\mathcal{U}_T, T\in \textup{RSYT}_\varsigma^{\pi}$ be a conformal block. 
    Let $j\in \llbracket 1, d-1 \rrbracket$. We have three kinds of asymptotics: \\
    
    $\bullet$ If $s_j=s_{j+1}=1$:
    \begin{align*}
      \lim_{x_{j+1}\to x_j}(x_{j+1}-x_j)^{\frac{1}{3}}\;\mathcal{U}_T(x_1,\dots, x_d)= \begin{cases}
            \mathcal{U}_{T'}(x_1,\dots,x_j,x_{j+2},\dots, x_d) \text{ if $j$ and $j+1$ are in different rows,}\\
            0 \text{ otherwise,}
        \end{cases}
    \end{align*}
    where $T'$ is the tableau obtained from $T$ by changing $j+1$ to $j$.\\

 $\bullet$ If $s_j=s_{j+1}=2$:

    \begin{align*}
         \lim_{x_{j+1}\to x_j}(x_{j+1}-x_j)^{\frac{1}{3}}\;\mathcal{U}_T(x_1,\dots, x_d) = \begin{cases}
            \mathcal{U}_{T'}(x_1,\dots,x_j,x_{j+2},\dots, x_d) \text{if $j$ and $j+1$ are not in the same two rows,}\\
            0 \text{ otherwise,}
        \end{cases}
    \end{align*}
    where $T'$ is the tableau obtained from $T$ by changing $j+1$ to $j$ and then removing $3$ boxes containing $j$ in $3$ different rows. \\
    
     $\bullet$ If $s_j\neq s_{j+1}$:
    
    \begin{align*}
         \lim_{x_{j+1}\to x_j}(x_{j+1}-x_j)^{\frac{2}{3}}\;\mathcal{U}_T(x_1,\dots, x_d)= \begin{cases}
            \mathcal{U}_{T'}(x_1,\dots,x_{j-1},x_{j+2},\dots, x_d) \text{ if $j$ and $j+1$ are in different rows,}\\
            0 \text{ otherwise,}
        \end{cases}
    \end{align*}
    where $T'$ is the tableau obtained from $T$ by removing boxes containing $j$ and $j+1$.
    
\end{proposition}
\begin{proof}
 We proceed case by case. \\
 
\textbf{The case $s_j = s_{j+1} = 1$.} We treat the two cases separately. If $j$, $j+1$ lie in different rows of $T\in \text{RSYT}_\varsigma^{\pi}$, then $T'\in \text{RSYT}_{\varsigma'}^{\pi}$ where $\varsigma'=(s_1,\dots,s_{j-1},2,s_{j+2},\dots,s_d)$. The proof then follows directly from Lemma \ref{lemmafusion} and Proposition \ref{propeqW3forUT}. On the other hand, if $j$ and $j+1$ are on the same row of $T$, it means that $T' \notin \text{RSYT}_{\varsigma'}^{\pi}$ and that
$$\mathcal U_T(x_1,\dots,x_d) = (x_{j+1}-x_j)^{\frac23} f(x_1,\dots,x_d),$$
where $f$ is analytic at $x_{j+1}=x_j$. This proves the second asymptotic property. \\

\textbf{The case $s_j = s_{j+1} = 2$.} Utilizing the representation \eqref{definitionCB} and omitting the variable dependencies of the conformal blocks we write
\begin{equation} \label{limR1R2}
\frac{\lim_{x_{j+1}\to x_j}(x_{j+1}-x_j)^{\frac{1}{3}}\;\mathcal{U}_T}{\mathcal{U}_{T'}} = \lim_{x_{j+1}\to x_j} \left[R_1 R_2\right],
\end{equation}
with
\begin{equation} \label{defR1R2}
    R_1 = \frac{(x_{j+1}-x_j)^{\frac23} \prod_{1\leq l<k\leq d} (x_k-x_l)^{-\frac{s_k s_l}{3}}}{\prod_{\substack{1\leq l<k\leq d \\ l,k \neq j+1}} (x_k-x_l)^{-\frac{s_k' s_l'}{3}}}, \qquad R_2 = \frac{\mathcal P_{T^t}}{(x_{j+1}-x_j) \mathcal P_{T'^t}}.
\end{equation}
We now show that the limit of $R_1$ and $R_2$ is well defined. Starting with $R_1$, we rewrite
$$\lim_{x_{j+1}\to x_j} R_1 = \lim_{x_{j+1}\to x_j} \prod_{k=1}^d \prod_{l=1}^{k-1} (x_j-x_i)^{-\frac{s_k s_l}{3}}\prod_{\substack{k=1 \\ k \neq j+1}}^d \prod_{\substack{l=1 \\ l \neq j+1}}^{k-1} (x_j-x_i)^{\frac{s_k' s_l'}{3}}.$$
In the first double product, we exclude all terms containing $x_{j+1}$:
$$\lim_{x_{j+1}\to x_j} R_1 = \lim_{x_{j+1}\to x_j} \prod_{l=1}^{j-1} (x_{j+1}-x_l)^{-\frac{2s_l}3} \prod_{k=j+2}^d (x_k-x_{j+1})^{-\frac{2s_k}3}\prod_{\substack{k=1 \\ k \neq j+1}}^d \prod_{\substack{l=1 \\ l \neq j+1}}^{k-1} (x_j-x_i)^{\frac{s_k' s_l'-s_k s_l}{3}}.$$
The only terms for which $s_k' s_l'-s_k s_l \neq 0$ are for $k=j$ or for $l=j$. Excluding all such terms in the last double product, we obtain 
\begin{equation}
    \lim_{x_{j+1}\to x_j} R_1 = \prod_{k=1}^{j-1} (x_j-x_k)^{-s_k} \prod_{k=j+2}^d (x_k-x_j)^{-s_k}.
\end{equation}
Finally, since $s_j=s_{j+1}=2$ and $T$ has three rows, there are two cases two consider. If $j$ and $j+1$ appear in the same two rows of $T$, then $j$ and $j+1$ appear in the same two columns of $T^t$, and we infer that $\mathcal P_{T^t}$ contains a factor $(x_{j+1}-x_j)^2$. Hence the limit of $R_2$ equals zero. On the other hand, if $j$ and $j+1$ do not appear in the same two rows of $T$, then there is one row (resp. column) of $T$ (respt. $T^t$) where both $j$ and $j+1$ appear, and we infer that $\mathcal P_{T^t}$ contains a factor $x_{j+1}-x_j$. Hence the limit is finite and nonzero. Finally, due to the structure of the Specht polynomials $\mathcal P_{T^t}$ as a product of three Vandermonde polynomials and to the form of $(T')^t$ relative to $T^t$, a careful inspection shows that 
\begin{equation}
    \lim_{x_{j+1}\to x_j} R_2 = \prod_{k=1}^{j-1} (x_j-x_k)^{s_k} \prod_{k=j+2}^d (x_k-x_j)^{s_k},
\end{equation}
which leads to the desired asymptotic result. \\

\textbf{The case $s_j \neq s_{j+1}$.}
In this case we assume that $s_j=1$ and $s_{j+1}=2$. The case $s_j=2$ and $s_{j+1}=1$ is treated in the same way. We start with
\begin{equation*} \label{limS1S2}
\frac{\lim_{x_{j+1}\to x_j}(x_{j+1}-x_j)^{\frac23}\;\mathcal{U}_T}{\mathcal{U}_{T'}} = \lim_{x_{j+1}\to x_j} \left[S_1 S_2\right],
\end{equation*}
with
$$S_1 = \frac{(x_{j+1}-x_j)^{\frac23} \prod_{1\leq l<k\leq d} (x_k-x_l)^{-\frac{s_k s_l}{3}}}{\prod_{\substack{1\leq l<k\leq d \\ l,k \neq j+1}} (x_k-x_l)^{-\frac{s_k' s_l'}{3}}}, \qquad S_2 = \frac{\mathcal P_{T^t}}{\mathcal P_{T'^t}}.$$
Computations that are similar to the case above lead to the limit
$$\lim_{x_{j+1}\to x_j} S_1 = \lim_{x_{j+1}\to x_j} R_1 = \prod_{k=1}^{j-1} (x_j-x_k)^{-s_k} \prod_{k=j+2}^d (x_k-x_j)^{-s_k}.$$
On the other hand, if $j$ and $j+1$ lie on a same row in $T$, it means that they lie on a same column in $T^t$. Therefore, in this case the polynomial $\mathcal P_{T^t}$ contains a factor $x_{j+1}-x_j$ and we infer that $\lim_{x_{j+1}\to x_j} S_2 =0$. Finally, if $j$ and $j+1$ lie on three different rows of $T$, then it is clear by inspection that 
$$\lim_{x_{j+1}\to x_j} S_2 = \lim_{x_{j+1}\to x_j} R_2 = \prod_{k=1}^{j-1} (x_j-x_k)^{s_k} \prod_{k=j+2}^d (x_k-x_j)^{s_k}.$$
This completes the proof.
\end{proof}

\appendix

\section{Informal derivation of the system of $d+8$ PDEs from CFT} 

In this Appendix we explain how the system of $d+8$ PDEs described in Section \ref{section2} arises in the CFT setting. We adopt the conformal bootstrap approach of \cite{BPZ} and therefore we do not rely on a specific CFT. Nevertheless, let us mention that some of the content of this Appendix has been made rigorous for the (real) $\mathfrak{sl}_3$ Toda CFT in \cite{C,CH2}.

\subsection{$W_3$ highest-weight modules and null-vectors}
Fix a parameter $c\in \mathbb C$, called the central charge. The $W_3$-algebra is the infinite-dimensional algebra generated by $\{L_n,\; W_n,\;n\in \mathbb Z\}$ subject to the following commutation relations \cite{Zam85}:

\begin{align*}
    & [L_n,L_m] = (n-m) L_{n+m} + \frac{c}{12} (n^3-n) \delta_{n,-m}, \\
    & [L_n,W_m] = (2n-m) W_{n+m}, 
\end{align*}
as well as
\begin{align*}
    [W_n,W_m] = \; & \frac{ c}{3\times5!} (n^2-1)(n^2-4) n \delta_{n,-m} + \frac{16}{22+5c} (n-m) \Lambda_{n+m} \\
    & + (n-m) \lb \frac{1}{15} (n+m+2)(n+m+3) - \frac16 (n+2)(m+2) \rb L_{n+m}, 
\end{align*}
where
\begin{align*}
    & \Lambda_n = \sum_{k \in \mathbb Z} : L_k L_{n-k} : + \frac15 x_n L_n, \\
    & x_{2l} = (1+l)(1-l), \\
    & x_{2l+1} = (2+l)(1-l),
\end{align*}
and where $::$ denotes the normal ordering, that is, $:L_{n_1} L_{n_2}: \; = L_{n_2} L_{n_1}$ if $n_1 > n_2$, or $:L_{n_1} L_{n_2}: \; = L_{n_1} L_{n_2}$ if $n_2 > n_1$.

Verma modules $V^{(\boldsymbol{\alpha},c)}$ over the $W_3$ algebra are labelled by vectors $\boldsymbol\alpha = \alpha_1 \boldsymbol\omega_1 + \alpha_2 \boldsymbol\omega_2$ which are linear combinations of fundamental weights of the Lie algebra $\mathfrak{sl}_3$ and by the central charge $c$. We have
\begin{align*}
    V^{(\boldsymbol{\alpha},c)} = \mathcal U(W_3)/I,
\end{align*}
where $I$ is the left ideal generated by $L_0-h$, $W_0-\rho$, $L_n$ and $W_n$ for $n>0$, and where $h$ and $\rho$ are expressed in terms of $\alpha_1,\alpha_2$ and $c$. In particular, $V^{(\boldsymbol{\alpha},c)}$ is generated by a highest-weight vector $v_0$ (that is, $V^{(\boldsymbol{\alpha},c)}= \mathcal U(W_3) \, v_0$) which satisfies $L_0 v_0 = h v_0$, $W_0 v_0 = \rho v_0$ and $L_n v_0 = W_n v_0 = 0$ for all $n>0$. Moreover, the Verma module $V^{(\boldsymbol{\alpha},c)}$ is graded by the $L_0$ eigenvalues. More precisely, we have
$$V^{(\boldsymbol{\alpha},c)} = \bigoplus_{n \geq 0} V^{(\boldsymbol{\alpha},c)}_n,$$
and elements $v \in V^{(\boldsymbol{\alpha},c)}_n$ satisfy $L_0 v = (h+n) v$. For this reason, any $v \in V^{(\boldsymbol{\alpha},c)}_n$ is said to be of level $n$.

A singular vector $w_n \in V^{(\boldsymbol{\alpha},c)}$ is a highest-weight vector of level $n>0$. If a non-zero singular vector $w_n \in V^{(\boldsymbol{\alpha},c)}$ can be found, then $V^{(\boldsymbol{\alpha},c)}$ is said to be degenerate and $w_n$ generates a proper submodule of $V^{(\boldsymbol{\alpha},c)}$. 

In this Appendix we consider highest-weight vectors $v_0^{(j)} \in V^{(i \beta \boldsymbol{\omega}_j,c)}$ which satisfy 
\begin{equation}
    L_0 v_0^{(j)} = h v_0^{(j)}, \qquad W_0 v_0^{(j)} = \rho^{(j)} v_0^{(j)},
\end{equation}
where $h$ is given in \eqref{relationcandbeta} and where
$$\rho^{(j)} = (-1)^j \frac{\left(3-5 \beta ^2\right) \left(3-4 \beta ^2\right)}{27 \beta}.$$

Each of $v_0^{(1)}$ and $v_0^{(2)}$ possesses three singular vectors. Define $\chi_i^{(j)} \in \mathcal U \lb W_3 \rb$ for $i=1,2,3$ and $j=1,2$ as follows:
\begin{align}
    \label{defchi1} & \chi_1^{(j)} := \lb W_{-1} -\frac{3\rho^{(j)}}{2h} L_{-1} \rb, \\
    \label{defchi2} & \chi_2^{(j)} := W_{-2} - \left( \frac{12\rho^{(j)}}{h(5h+1)} L_{-1}^2 - \frac{6\rho^{(j)}(h+1)}{h(5h+1)} L_{-2} \right), \\
    \label{defchi3} & \chi_3^{(j)} := W_{-3} -  \left( \frac{16\rho^{(j)}}{h(h+1)(5h+1)} L_{-1}^3 - \frac{12\rho^{(j)}}{h(5h+1)} L_{-1} L_{-2} - \frac{3\rho^{(j)}(h-3)}{2h(5h+1)} L_{-3} \right).
\end{align}
Then, the null-vectors are $\chi_i^{(j)} v_0^{(j)}$ for $i=1,2,3$ and $j=1,2$.

\subsection{CFT derivation of the PDEs}

In this section we explain how the PDEs \eqref{W3BPZ} and \eqref{ward1}--\eqref{ward5} arise in CFTs with an extended $W_3$-symmetry algebra. The chiral part of the algebra of symmetries of the CFT is encoded into two local fields $T(z)$ and $W(z)$ which are defined as follows:
\begin{align*}
  &  T(z) = \sum_{n \in \mathbb Z} \frac{L_n}{z^{n+2}}, \qquad \mathcal W(z) = \sum_{n \in \mathbb Z} \frac{W_n}{z^{n+3}},
\end{align*}
where the Laurent modes $L_n$ and $W_n$ satisfy the commutation relations of the $W_3$-algebra. Important local fields in such CFTs are the primary fields $V_{\boldsymbol\alpha}$ which correspond to highest-weight vectors in $W_3$ highest-weights modules, as well as their descendants $L_{-n} V_{\boldsymbol\alpha}$ and $W_{-n} V_{\boldsymbol\alpha}$. We now define the correlation function
\begin{equation} \label{defFvarsigma}
    \mathcal F_\varsigma(x_1,\dots,x_d) := \left\langle \prod_{j=1}^d V_{i\beta \boldsymbol{\omega}_{s_j}}(x_j) \right\rangle.
\end{equation}
Crucial identities in CFT are the local Ward identities which measure the effect of inserting the fields $T(z)$ and $\mathcal W(z)$ inside correlation functions. We have \cite{FL}
\begin{align} \label{localwardforT}
     \left\langle T(z) \prod_{j=1}^d V_{i\beta \boldsymbol{\omega}_{s_j}}(x_j) \right\rangle = \sum_{i=1}^d \left( \frac{h}{(z-x_i)^2} + \frac{\partial_{x_i}}{z-x_i} \right) \mathcal F_\varsigma(x_1,\dots,x_d),
\end{align}
and
\begin{equation} \label{localwardforW}
    \left\langle \mathcal W(z) \prod_{j=1}^d V_{i\beta \boldsymbol{\omega}_{s_j}}(x_j) \right\rangle = \sum_{i=1}^d \left( \frac{\rho^{(s_i)}}{(z-x_i)^3} + \frac{W_{-1}^{(i)}}{(z-x_i)^2} + \frac{W_{-2}^{(i)}}{z-x_i} \right) \mathcal F_\varsigma(x_1,\dots,x_d),
\end{equation}
where 
\begin{align*} 
   W_{-n}^{(i)} \mathcal F_\varsigma(x_1,\dots,x_d) :=  \left\langle \lb W_{-n} V_{i\beta \boldsymbol{\omega}_{s_i}} \rb(x_i) \prod_{j \neq i}^d V_{i\beta \boldsymbol{\omega}_{s_j}}(x_j) \right\rangle.
\end{align*} 
Starting from the local Ward identities \eqref{localwardforT} and \eqref{localwardforW}, standard CFT calculations following \cite[Section 6.6.1]{DMS} can be performed to express the correlation function of $W_3$-descendant fields in terms of correlation functions of the corresponding primary fields. For instance, for $n \in \mathbb Z$ we have
\begin{align} \label{corrdescendantLn}
L_{-n}^{(i)} \mathcal F_\varsigma(x_1,\dots,x_d) & := \left\langle \lb L_{-n} V_{i\beta \boldsymbol{\omega}_{s_i}} \rb(x_i) \prod_{j \neq i}^d V_{i\beta \boldsymbol{\omega}_{s_j}}(x_j) \right\rangle \\
\nonumber & = \sum_{j \neq i}^d \lb \frac{(n-1)h}{(x_j-x_i)^n} - \frac{\partial_{x_j}}{(x_j-x_i)^{n-1}} \rb \mathcal F_\varsigma(x_1,\dots,x_d),
\end{align}
as well as
\begin{equation} \label{corrdescendantWn}
    \left\langle \lb W_{-n} V_{i\beta \boldsymbol{\omega}_{s_i}} \rb(x_i) \prod_{j\neq i}^d V_{i\beta \boldsymbol{\omega}_{s_j}}(x_j) \right\rangle = \sum_{j\neq i}^d \left(-\frac{(n-1)(n-2)\rho^{(s_j)}}{2(x_j-x_i)^n} + \frac{(n-2)W_{-1}^{(j)}}{(x_j-x_i)^{n-1}} - \frac{W_{-2}^{(j)}}{(x_j-x_i)^{n-2}} \right)  \mathcal F_\varsigma(x_1,\dots,x_d).
\end{equation}
Identities \eqref{corrdescendantLn} and \eqref{corrdescendantWn} can be generalized to more intricate $W_3$-descendants in a straightforward way. 

We are now ready to derive Equations \eqref{W3BPZ} and \eqref{ward1}--\eqref{ward5}. The null-vectors \eqref{defchi1}, \eqref{defchi2} and \eqref{defchi3} provide a set of $3d$ constraints satisfied by $\mathcal F_\varsigma$. More precisely, for $m=1,\cdots,d$ we have
\begin{align}
  \label{actionW1}  & \lb W_{-1}^{(m)} -\frac{3\rho^{(s_m)}}{2h} L_{-1}^{(m)} \rb \mathcal F_\varsigma = 0, \\
   \label{actionW2} & \lb W_{-2}^{(m)} - \left( \frac{12\rho^{(s_m)}}{h(5h+1)} \lb L_{-1}^{(m)}\rb^2 - \frac{6\rho^{(s_m)}(h+1)}{h(5h+1)} L_{-2}^{(m)} \right) \rb \mathcal F_\varsigma = 0, \\
   \label{actionW3} & \lb W_{-3}^{(m)} - \left( \frac{16\rho^{(s_m)}}{h(h+1)(5h+1)} \lb L_{-1}^{(m)}\rb^3 - \frac{12\rho^{(s_m)}}{h(5h+1)} L_{-1}^{(m)} L_{-2}^{(m)} - \frac{3\rho^{(s_m)}(h-3)}{2h(5h+1)} L_{-3}^{(m)} \right)\rb \mathcal F_\varsigma = 0.
\end{align}
Next, substitution of \eqref{corrdescendantWn} for $n=3$ into \eqref{actionW3} leads to
\begin{align}
   \label{zgueg} & \bigg[ \sum_{i \neq m}^d \left(-\frac{\rho^{(s_i)}}{(x_i-x_m)^3} + \frac{W_{-1}^{(i)}}{(x_i-x_m)^2} - \frac{W_{-2}^{(i)}}{(x_i-x_m)} \right) \\
  \nonumber  & - \left( \frac{16\rho^{(s_m)}}{h(h+1)(5h+1)} \lb L_{-1}^{(m)}\rb^3 - \frac{12\rho^{(s_m)}}{h(5h+1)} L_{-1}^{(m)} L_{-2}^{(m)} - \frac{3\rho^{(s_m)}(h-3)}{2h(5h+1)} L_{-3}^{(m)} \right) \bigg] \mathcal F_\varsigma = 0.
\end{align}
It finally remains to  substitute \eqref{actionW1} and \eqref{actionW2} into \eqref{zgueg}, use \eqref{corrdescendantLn} as well as the parameter correspondence
\begin{align} \label{correspondencerhoq}
    \rho^{(s_i)}= -q_i \frac{\left(3-5 \beta ^2\right) \left(3-4 \beta ^2\right)}{27 \beta}
\end{align}
to obtain \eqref{W3BPZ}. Finally, the knowledge of the identity \eqref{localwardforW} and the fact that $W(z) = O(z^{-6})$ as $|t|\to \infty$ \cite{FL} implies that $\mathcal F_\varsigma$ satisfies
\begin{align}
  \label{otherformward1}  & \sum_{i=1}^d W_{-2}^{(i)} \mathcal F_\varsigma = 0, \\
  \label{otherformward2}  & \sum_{i=1}^d \lb x_i W_{-2}^{(i)} + W_{-1}^{(i)} \rb \mathcal F_\varsigma = 0, \\
   \label{otherformward3} & \sum_{i=1}^d \lb x_i^2 W_{-2}^{(i)} + 2 x_i W_{-1}^{(i)} + \rho^{(s_i)} \rb \mathcal F_\varsigma = 0, \\
   \label{otherformward4} & \sum_{i=1}^d \lb x_i^3 W_{-2}^{(i)} + 3 x_i^2 W_{-1}^{(i)} + 3 x_i \rho^{(s_i)} \rb \mathcal F_\varsigma = 0, \\
   \label{otherformward5} & \sum_{i=1}^d \lb x_i^4 W_{-2}^{(i)} + 4 x_i^3 W_{-1}^{(i)} + 6 x_i^2 \rho^{(s_i)} \rb \mathcal F_\varsigma = 0.
\end{align}
Using \eqref{correspondencerhoq}, substitution of \eqref{actionW1} and \eqref{actionW2} into \eqref{otherformward1}, \eqref{otherformward2}, \eqref{otherformward3}, \eqref{otherformward4} and \eqref{otherformward5} yield \eqref{ward1}, \eqref{ward2}, \eqref{ward3}, \eqref{ward4} and \eqref{ward5}.

\newcommand{\changeurlcolor}[1]{\hypersetup{urlcolor=#1}}      
\changeurlcolor{black}

\end{document}